\documentclass[sigconf]{acmart}
\AtBeginDocument{%
  }

\setcopyright{acmlicensed}
\copyrightyear{2026}
\acmYear{2026}
\acmDOI{XXXXXXX.XXXXXXX}
\acmConference[KDD '26]{KDD}{August 09--13,
  2026}{Jeju, Korea}
\acmISBN{978-1-4503-XXXX-X/2018/06}

\usepackage[utf8]{inputenc} 
\usepackage[T1]{fontenc}    
\usepackage{hyperref}       
\usepackage{url}            
\usepackage{booktabs}       
\usepackage{amsfonts}       
\usepackage{amsmath}
\usepackage{nicefrac}       
\usepackage{microtype}      
\usepackage{xcolor}         
\usepackage{bm}
\usepackage{graphicx}
\usepackage{fancyvrb}
\usepackage{fancyhdr}
\usepackage{enumitem}
\usepackage{algorithm}
\usepackage{tikz}
\usepackage{algorithmic}
\newtheorem{theorem}{Theorem}[section]
\newtheorem{lemma}[theorem]{Lemma}
\newtheorem{example}[theorem]{Example}
\newtheorem{proposition}[theorem]{Proposition}
\newtheorem{corollary}[theorem]{Corollary}

\usepackage{adjustbox}

\DeclareMathOperator{\A}{\mathbf{A}}

\DeclareMathOperator{\W}{\mathbf{W}}
\DeclareMathOperator{\y}{\mathbf{y}}
\DeclareMathOperator{\bL}{\mathbf{L}}
\DeclareMathOperator{\blam}{\bm{\lambda}}
\DeclareMathOperator{\yhat}{\hat{\mathbf{y}}}
\DeclareMathOperator{\hier}{\mathcal{H}}
\DeclareMathOperator{\lagr}{\mathcal{L}}
\DeclareMathOperator{\Q}{\mathbf{Q}}
\DeclareMathOperator{\zeros}{\mathbf{0}}
\DeclareMathOperator{\diag}{diag}
\def\lsqr{\text{LSQR}}
\def\inv{^{-1}}

\usepackage{enumitem}
\setlist[itemize]{leftmargin=12pt}
\setlist[enumerate]{leftmargin=14pt}

\graphicspath{ {./figures/} }

\begin{document}

\title{Billions-Scale Forecast Reconciliation}


\author{Tianyu Wang}
\email{twang147@jhu.edu}
\affiliation{
\institution{Johns Hopkins University}
\city{Baltimore}
\state{Maryland}
\country{USA}
}
\authornote{Work completed while interning at Amazon in the summer of 2025.}
\author{Matthew C. Johnson}
\email{jnmmatt@amazon.com}
\affiliation{
    \institution{Amazon}
    \city{Seattle}
    \state{Washington}
    \country{USA}
}
\authornote{Long Term Planning and Forecasting, Supply Chain Optimization Technologies}
\author{Steven Klee}
\authornotemark[2]
\email{sklee@amazon.com}
\affiliation{
    \institution{Amazon}
    \city{Seattle}
    \state{Washington}
    \country{USA}
}
\author{Matthew L. Malloy}
\authornotemark[2]
\email{mlmalloy@amazon.com}
\affiliation{
    \institution{Amazon}
    \city{Seattle}
    \state{Washington}
    \country{USA}
}

\renewcommand{\shortauthors}{Wang et al.}

\begin{abstract}
    The problem of combining multiple forecasts of related quantities that obey expected equality and additivity constraints, often referred to a hierarchical forecast reconciliation, is naturally stated as a simple optimization problem.  In this paper we explore optimization-based point forecast reconciliation at scales faced by large retailers.  We implement and benchmark several algorithms to solve the forecast reconciliation problem, showing efficacy when the dimension of the problem exceeds \emph{four billion} forecasted values.  To the best of our knowledge, this is the largest forecast reconciliation problem, and perhaps on-par with the largest constrained least-squares-problem ever solved.  We also make several theoretical contributions.  We show that for a restricted class of problems and when the loss function is weighted appropriately, least-squares forecast reconciliation is equivalent to \emph{share-based} forecast reconciliation.  This formalizes how the optimization based approach can be thought of as a generalization of share-based reconciliation, applicable to multiple, overlapping data hierarchies.
\end{abstract}


\received{8 February 2026}
\received[revised]{XYZ}
\received[accepted]{XYZ}

\maketitle

%

\section{Introduction}\label{sec:intro}
Corporations have a need to create unified plans across their businesses.  Central to these plan are \emph{coherent} time-series forecasts of key business metrics (e.g., sales).  \emph{Coherent} forecasts---forecasts that aggregate as expected by data taxonomies---are often not a given, for two reasons.  First, accurate forecasts require intimate business knowledge; the data, models, and constituent forecasts are often owned by disparate teams.  Second, responsibilities of centralized forecasting teams are continuously evolving to include forecasts of new dimensions (e.g., sales broken out by region or by a new product attribute), and multivariate econometric models struggle to scale.  To maintain separation of concerns and ensure scalability, it can be advantageous to enforce coherence as a post-processing step.  In large retail corporations, this problem can be extreme, involving a large number of teams, extensive regional, product, seller, and customer data hierarchies with overlapping definitions, and billions of forecasted quantities.

Forecast reconciliation has received considerable attention over the last decade, and it has been shown that there are strong advantages over classical top-down and bottom-up forecasting \cite{athanasopoulos}.  Forecast reconciliation can be stated as follows: given an initial collection of forecasts $\yhat$, find a new collection of forecasts $\mathbf{y}$ that is close to the original collection of forecasts which also satisfies an additive aggregation constraint $\mathbf{A} \mathbf{y}=\mathbf{0}$. These aggregations may be defined by the business and data definitions. When forecasting retail demand, a non-negativity constraint $\mathbf{y} \geq \mathbf{0}$ is often included. If ``closeness'' of forecasts is defined by weighted euclidean distance, forecast reconciliation can be expressed as a simple optimization problem:  
\begin{equation}
\begin{aligned} \label{eqn:basic}
\min_{\mathbf{y} } \quad & \frac{1}{2}\| \mathbf{y} - \mathbf{\widehat{y}} \|_{2,{\mathbf{W}}}^2\\
\textrm{s.t.} \quad & \mathbf{A} \mathbf{y} = \mathbf{0} \\
  &\mathbf{y} \geq \mathbf{0}    \\
\end{aligned}
\end{equation}
where $\mathbf{A}$ is a matrix with entries from $\{0,1,-1\}$ that encodes hierarchical information (e.g, the regional forecasts must sum to the country-level forecasts, daily forecasts for January must sum to a forecast for the whole month of January, etc.) and $\mathbf{W}$ is a diagonal matrix of weights. It is important to allow for weighting because different time series may exist on vastly different scales, especially at high levels of business aggregation, and the weighting helps normalize changes across scales. As an example, changing the forecast of total sales for a large, multinational company by $\$1000$ does not have the same business meaning as changing the forecast for an individual item that averages $\$5000$ in sales per month by the same absolute amount. 

Absent the non-negativity constraint, Eq.~\eqref{eqn:basic} is a least squares problem with a readily-derived closed form solution.  Solving Eq.~\eqref{eqn:basic} with the non-negativity constraint requires iterative approaches, as the problem is an instance of cone regression \cite{dimiccoli2016fundamentalsconeregression} (or more generally a quadratic program). \emph{For any reasonably-sized problem}, as the objective is strongly convex, general-purpose convex optimization solvers will converge to the unique global minimum.  

In this paper we make a number of contributions.  First, we implement customized methods for efficiently solving Eq.~\eqref{eqn:basic} when $\widehat{\mathbf{y}}$ has more than \emph{four billion entries}.  We reconcile real-world demand data over a joint temporal and business hierarchy from a large retailer. In the time domain we forecast at the daily, weekly, and monthly levels. In the product catalog we consider low-level forecasts of individual items which are aggregated to intermediate levels including product category, product group, sort type, and total sales (see Section \ref{sec:data_description} for definitions and more details of the different product groupings). Each entry of $\yhat$ represents a forecast of a temporal and business entry at some point in the future at one of the above granularities. In this hierarchy, we solve Eq.~\eqref{eqn:basic} when
$$\mathbf{y} \in \mathbb{R}^{4,185,173,500} \mbox{ and } \mathbf{A} \in \mathbb{R}^{47,213 \times 4,185,173,500}.$$  

Because the temporal and product hierarchies overlap nontrivially, the problem in Eq.~ \eqref{eqn:basic} cannot be separated into smaller sub-problems.  To solve this problem at billions-scale we customize and implement several algorithms: Alternating Projections, Dykstra's Algorithm and Alternating Direction Method of Multiplies.  Our implementations rely on readily available sparse linear algebra routines and we use high-memory cloud computing instances. \emph{To the best of our knowledge, this is the largest forecast reconciliation problem ever solved (see \cite{SPRANGERS20241689}), and on-par with the largest least-squares problems ever solved}.  We describe details of implementation and stopping criteria in Section \ref{sec:algo}.

This paper also makes a theoretical contribution with important practical implications.  When the weights are assigned as $1/\yhat_i$ and the aggregation matrix $\A$ has \emph{disjoint row supports} (for example, in a simple hierarchy with two levels), the least squares solution is guaranteed to be nonnegative (see Theorem \ref{thm:nonneg}).  For more general tree-based hierarchies, we show that top-down (share-based) disaggregation and bottom-up aggregation are limiting solutions of certain \emph{top-heavy} and \emph{bottom-heavy} weighting schemes (see Theorem \ref{thm:top_bottom_heavy_aggs}). This is important as it means the optimization-based approach can be viewed as a flexible extension of traditional top-down and bottom-up forecasting.  It also implies that including the non-negativity constraint in Eq.~\eqref{eqn:basic} is redundant in the sense that the solution to the unconstrained optimization is \emph{always} non-negative in these settings.  

To summarize, in this paper we paper make the following contributions:
\begin{enumerate}
\item We formulate the reconciliation problem in a way that allows simultaneous reconciliation of any number of overlapping sets of forecasts and data hierarchies, more in-line with modern optimization literature.
\item We demonstrate an efficient approach to generate the constraint matrix $\mathbf{A}$ from standard tabular datasets. See Section \ref{sec:constraints}.
\item We implement and benchmark several algorithms to solve extensions of Eq.~\eqref{eqn:basic} at billions-scale.  See Sections \ref{sec:algo} and \ref{sec:exp}.
\item We show that Eq.~\eqref{eqn:basic} recovers \emph{share-based} reconciliation for a restricted class of problems  provided the weights on the second level are set as $1/\widehat{y_i}$.  See Section \ref{sec:discussion}.  This means that the optimization approach can be thought of as an extension of share-based (top-down) reconciliation. 
\end{enumerate}

This paper is organized as follows. In Section \ref{sec:algo} we formulate the problem and introduce the algorithms we use to solve Eq.~\eqref{eqn:basic}.  We describe experiments and compare performance in Section \ref{sec:exp}. In Section \ref{sec:discussion} we discuss the results and relate the optimization framework to share-based top-down forecasting. We list related work in Section \ref{sec:related}.

\section{Formulation and Algorithms}\label{sec:algo}
We focus on the point forecast optimization reconciliation in Eq.~\eqref{eqn:basic}.  This is a constrained least-squares problem with linear (in)equality constraints, which is a specific form of a cone regression and a more general quadratic program. Under some conditions, it can be formulated as a non-negative least squares problem, but in general cannot (for example when $\mathbf{A}$ corresponds to multiple overlapping hierarchies). Absent the non-negativity constraints, Eq.~\eqref{eqn:basic} has a readily derived closed form solution:
\begin{equation}\label{eq:lsqr}
\mathbf{y}^*_{\text{LSQR}}=\mathbf{\widehat{y}}-\mathbf{W}^{-1}\mathbf{A}^\top (\mathbf{A}\mathbf{W}^{-1}\mathbf{A}^\top)^{-1}\mathbf{A}\mathbf{\widehat{y}}.
\end{equation}
When the non-negativity constraint $\mathbf{y} \geq \mathbf{0}$ is included, iterative methods are required.  Although there are many general purpose convex optimization solvers in Python that can be used to solve Eq.~\eqref{eqn:basic}, they do not scale to our use case. Solving this optimization problem for forecast reconciliation is challenging in practice because the size of $\mathbf{y}$ is extremely large and the entries in $\mathbf{y}$ vary by orders of magnitude. Consequently:
\begin{itemize}
\item Solvers using interior point algorithm such as \texttt{Clarabel}, \texttt{CVXOPT}, and \texttt{ECOS} are sensitive to poor scaling \cite{wright}. When we use them to reconcile our forecasts, they fail to converge. We demonstrate this with an example in Section \ref{sec:exp}.
\item Solvers such as \texttt{DAQP} and \texttt{qpOASES} are designed for dense matrices, and thus can be computationally expensive in our case when the matrices are large and sparse \cite{arnstrom,stark}.  
\item Solvers such as \texttt{Gurobi} and \texttt{MOSEK} require commercial licenses \cite{gurobi,mosek}.
\end{itemize}

After exploring different algorithms, we propose the following three algorithms.

\subsection{Alternating Projections}\label{sec:algo_alter}

In our problem, with the non-negativity constraints, the optimized solution lies in the intersection of two convex sets $\mathbf{A}\mathbf{y}=\mathbf{0}$ and $\mathbf{y}\geq \mathbf{0}$. Therefore, one natural solution is to use the alternating projection algorithm, which finds a point in this intersection by iteratively projecting onto each of the set. Even with large sparse $\mathbf{A}$, the projections can be efficient, and minimally affected by the scale.

For this algorithm, we apply a stopping criterion that considers two types of residuals:
\begin{itemize}
\item iterate change in $\mathbf{y}$: $r_{\text{iter}}=\|\mathbf{y}_{(t+1)}-\mathbf{y}_{(t)}\|\leq\epsilon_{\text{iter}}$
\item feasibility residuals: $r_{\text{fea}}=\|\mathbf{A}\mathbf{y}_{(t+1)}\|\leq\epsilon_{\text{fea}}$ or \\  $r_{\text{fea}}=\|(\mathbf{y}_{(t+1)})_{-}\|\leq\epsilon_{\text{fea}}$, depending on which projection is done first at each iteration
\end{itemize}

Algorithm \ref{algo_alter} describes the alternating projection approach in detail. The last step of each iteration is to project onto the null space of $\mathbf{A}$, so each $\mathbf{y}_{(t)}$ is guaranteed to satisfy $\mathbf{A}\mathbf{y}_{(t)}=\mathbf{0}$. Thus, the stopping criteria only checks the feasibility violation $r_{\text{fea}}=\|(\mathbf{y}_{(t+1)})_{-}\|$.

\begin{algorithm}
    \caption{Alternating Projection algorithm for Optimization Problem Eq.~\eqref{eqn:basic}}
    \label{algo_alter}
    \begin{algorithmic}
        \STATE Initialize $\mathbf{y}_{(0)}=\mathbf{\widehat{y}}$. Set $t=0$, $r_{\text{iter}}=r_{\text{fea}}=1e^{10}$. Choose small values $\epsilon_{\text{iter}},\epsilon_{\text{fea}}$.
        \STATE Pre-compute $\mathbf{W}^{-1}\mathbf{A}^\top$ and $(\mathbf{A}\mathbf{W}^{-1}\mathbf{A}^\top)^{-1}$.
        \WHILE{$r_{\text{iter}}>\epsilon_{\text{iter}}$ or $r_{\text{fea}}>\epsilon_{\text{fea}}$}
            \STATE $\mathbf{y}_{(t+1)}\leftarrow (\mathbf{y}_{(t)})_+$, i.e. project onto non-negative orthant;
            \STATE $\mathbf{y}_{(t+1)}\leftarrow \mathbf{y}_{(t+1)}-\mathbf{W}^{-1}\mathbf{A}^\top \big((\mathbf{A}\mathbf{W}^{-1}\mathbf{A}^\top)^{-1}(\mathbf{A}\mathbf{y}_{(t+1)})\big)$, i.e. project onto null space of $\mathbf{A}$;
            \STATE $r_{\text{iter}}\leftarrow\|\mathbf{y}_{(t+1)}-\mathbf{y}_{(t)}\|$; $r_{\text{fea}}\leftarrow\|(\mathbf{y}_{(t+1)})_{-}\|$;
            \STATE $t\leftarrow t+1$.
        \ENDWHILE
        \STATE Return $\mathbf{y}^*_{\text{AP}}:=\mathbf{y}_{(t)}$.
    \end{algorithmic}
\end{algorithm}

In our empirical experiment presented in Section \ref{sec:exp}, the algorithm generates a solution close to the true optimal within a short period of time; unfortunately there is no guarantee that this method converges to the optimal, as the non-negative orthant is not a linear subspace (see, for example, \cite{pang}).

\subsection{Dykstra's Algorithm}\label{sec:algo_dykstra}

Dykstra's algorithm, a variant of the alternating projection method, guarantees the convergence to the true projection \cite{dykstra}. The approach is detailed in Algorithm \ref{algo_dykstra}. 

\begin{algorithm}
    \caption{Dykstra's algorithm for Optimization Problem Eq.~\eqref{eqn:basic}}
    \label{algo_dykstra}
    \begin{algorithmic}
        \STATE Initialize $\mathbf{p}=\mathbf{q}=\mathbf{0}$, $\mathbf{y}_{(0)}=\mathbf{\widehat{y}}$. Set $t=0$, $r_{\text{iter}}=r_{\text{fea}}=1e^{10}$. Choose small values $\epsilon_{\text{iter}},\epsilon_{\text{fea}}$.
        \STATE Pre-compute $\mathbf{W}^{-1}\mathbf{A}^\top$ and $(\mathbf{A}\mathbf{W}^{-1}\mathbf{A}^\top)^{-1}$.
        \WHILE{$r_{\text{iter}}>\epsilon_{\text{iter}}$ or $r_{\text{fea}}>\epsilon_{\text{fea}}$}
            \STATE $\mathbf{u}\leftarrow \mathbf{y}_{(t)}+\mathbf{p}$;
            \STATE $\mathbf{y}_{(t+1)}\leftarrow (\mathbf{u})_+$, i.e. project onto non-negative orthant;
            \STATE $\mathbf{p}\leftarrow \mathbf{u}-\mathbf{y}_{(t+1)}$;
            \STATE $\mathbf{v}\leftarrow \mathbf{y}_{(t+1)}+\mathbf{q}$;
            \STATE $\mathbf{y}_{(t+1)}\leftarrow \mathbf{v}-\mathbf{W}^{-1}\mathbf{A}^\top \big((\mathbf{A}\mathbf{W}^{-1}\mathbf{A}^\top)^{-1}(\mathbf{A}\mathbf{v})\big)$, i.e. project onto null space of $\mathbf{A}$;
            \STATE $\mathbf{q}\leftarrow \mathbf{v}-\mathbf{y}_{(t+1)}$;
            \STATE $r_{\text{iter}}\leftarrow\|\mathbf{y}_{(t+1)}-\mathbf{y}_{(t)}\|$;
            $r_{\text{fea}}\leftarrow\|(\mathbf{y}_{(t+1)})_{-}\|$;
            \STATE $t\leftarrow t+1$.
        \ENDWHILE
        \STATE Return $\mathbf{y}^*_{\text{Dykstra}}:=\mathbf{y}_{(t)}$.
    \end{algorithmic}
\end{algorithm}

\subsection{Alternating Direction Method of Multipliers (ADMM)}\label{sec:algo_admm}

Although Dykstra's algorithm converges to the projection even when the interacting sets are not linear, it lacks theoretical guarantee for the objective optimality under commonly used stopping criteria \cite{birgin}. 

In general, augmented Lagrangian algorithms, such as the alternating direction method of multipliers (ADMM), excel in optimization problems like ours. They allow problem splitting and efficient projections, making it helpful for large-scale problems \cite{parikh}. In addition, the penalty term regularizes the system and reduces sensitivity to bad scaling \cite{boyd}. Although ADMM-based solvers exist in Python packages such as \texttt{OSQP} and \texttt{SCS}, they are not tailored for our problem. After our empirical experiments, we found it much faster to implement ADMM on our problem by ourselves---we can split the problem and do projection separately based on our specific constraints.

In detail, we re-formulate Optimization Problem Eq.~\eqref{eqn:basic} as:
\begin{equation}\label{eqn:admm}
\begin{aligned}
\min_{\mathbf{y},\mathbf{z}} \quad (\mathbf{y} - \mathbf{\widehat{y}})^\top&\mathbf{W}(\mathbf{y}-\mathbf{\widehat{y}})+f(\mathbf{z})\\
\textrm{s.t.} \quad & \mathbf{A} \mathbf{y} = \mathbf{0} \\
  &\mathbf{y} = \mathbf{z}    \\
\end{aligned}
\end{equation}
where $f(\mathbf{z})=0$ if $\mathbf{z}\geq \mathbf{0}$ and $f(\mathbf{z}) = \infty$ otherwise. Based on this formulation, we can use the scaled form of ADMM algorithm to get the optimal solution; detailed steps are presented in Algorithm \ref{algo_admm}. 

\begin{algorithm}
    \caption{ADMM algorithm for Optimization Problem Eq.~\eqref{eqn:admm}}
    \label{algo_admm}
    \begin{algorithmic}
        \STATE Initialize $\mathbf{z}_{(0)}=\mathbf{\widehat{y}}$, $\mathbf{u}_{(0)}=\mathbf{0}$, $t=0$, $r_{\text{primal}}=r_{\text{dual}}=1$, $\epsilon_{\text{primal}}=\epsilon_{\text{dual}}=0$.
        \STATE Choose a value for $\rho$, small positive $\epsilon_{\text{abs}}$ and $\epsilon_{\text{rel}}$.
        \STATE Pre-compute matrix operations that are fixed across iterations.
        \WHILE{$r_{\text{primal}}>\epsilon_{\text{primal}}$ or $r_{\text{dual}}>\epsilon_{\text{dual}}$}
            \STATE $\mathbf{y}_{(t+1)}\leftarrow \arg\min_\mathbf{y}\{(\mathbf{y} - \mathbf{\widehat{y}})^\top\mathbf{W}(\mathbf{y}-\mathbf{\widehat{y}})+\frac{\rho}{2}\|\mathbf{y}-\mathbf{z}_{(t)}+\mathbf{u}_{(t)}\|_2^2\ \text{s.t. }\mathbf{A}\mathbf{y}=\mathbf{0}\}$;
            \STATE $\mathbf{z}_{(t+1)}\leftarrow (\mathbf{y}_{(t+1)}+\mathbf{u}_{(t)})_{+}$;
            \STATE $\mathbf{u}_{(t+1)}\leftarrow \mathbf{u}_{(t)}+\mathbf{y}_{(t+1)}-\mathbf{z}_{(t+1)}$;
            \STATE $r_{\text{primal}}\leftarrow\|\mathbf{y}_{(t+1)}-\mathbf{z}_{(t+1)}\|_2$; $r_{\text{dual}}\leftarrow\|\rho(\mathbf{z}_{(t+1)}-\mathbf{z}_{(t)})\|_2$;
            \STATE $\epsilon_{\text{primal}}\leftarrow\sqrt{n}\epsilon_{\text{abs}}+\epsilon_{\text{rel}}\cdot\max(\|\mathbf{y}_{(t+1)}\|_2,\|\mathbf{z}_{(t+1)}\|_2)$; $\epsilon_{\text{dual}}\leftarrow\sqrt{n}\epsilon_{\text{abs}}+\epsilon_{\text{rel}}\cdot\|\rho \mathbf{u}_{(t+1)}\|_2$;
            \STATE $t\leftarrow t+1$.
        \ENDWHILE
        \STATE Return $\mathbf{y}^*_{\text{ADMM}}:=\mathbf{y}_{(t)}$.
    \end{algorithmic}
\end{algorithm}
Note that the ``update $\mathbf{y}_{(t+1)}$'' step (i.e. first line inside the while loop) is a simple quadratic programming with equality constraints, which can be solved via KKT, i.e.
\begin{equation*}
\begin{bmatrix}
    \mathbf{y}_{(t+1)}\\
    \lambda
\end{bmatrix}
=
\begin{bmatrix}
    \mathbf{H} & \mathbf{A}^\top\\
    \mathbf{A} & \mathbf{0}
\end{bmatrix}^{-1}   
\begin{bmatrix}
    -\mathbf{c}\\
    \mathbf{0}
\end{bmatrix}
\end{equation*}
with $\mathbf{H}=2\mathbf{W}+\rho\mathbf{I}$ and $\mathbf{c}=-2\mathbf{W}\widehat{\mathbf{y}}+\rho(\mathbf{u}_{(t)}-\mathbf{z}_{(t)})$. Some matrix operations, such as $\mathbf{H}$ and the first term in $\mathbf{c}$, are fixed across iterations, and thus can be pre-computed for efficiency.

In our case, the quadratic term in Optimization Problem \eqref{eqn:admm} is a closed, proper, and convex function when $\mathbf{W}$ is positive semi-definite, and $f(\cdot)$ is closed, proper, and convex. Moreover, our constraints are linear, and the Slater condition holds as the constraint set is nonempty, so the unaugmented Lagrangian for Problem \eqref{eqn:admm} has a saddle point. Boyd et al.~\cite{boyd} prove that under these conditions, the ADMM iterates satisfy residual convergence, objective convergence, and dual variable convergence.

Although in practice, ADMM can be slow to generate a high accuracy solution, it usually converges to modest accuracy within a few tens of iterations \cite{odonoghue,stellato}, which is sufficient in our case. In fact, from empirical experiments we found that the converging speed is highly dependent on $\rho$ and the scale of the problem. Although we might need trial and error to find a nice $\rho$ in practice, once we find it, we can use it for all similar reconciliation tasks.

For the stopping criterion, we propose to check the primal and dual residuals, as it is shown in \cite{boyd} that the objective suboptimality can be bounded by their combination. As they suggest, we define the stopping criteria as described in Algorithm \ref{algo_admm}.

\section{Experiments}\label{sec:exp}
\subsection{Data Description and Definitions} \label{sec:data_description}

We complete a set of experiments reconciling forecasts of retail demand (specifically, the number of units sold) for a large retailer.  We use five sets of forecast datasets produced by distinct teams.  The forecasts consist of 
\begin{enumerate}
    \item weekly forecasts of individual items in the product catalog at the SKU level (e.g., Bounty Quick Size Paper Towels, White, 8-Family-Rolls);
    \item daily forecasts at the \emph{product group} level (e.g., all paper towels);
    \item monthly forecasts at the \emph{product family} level (e.g., all consumable items, which includes paper towels);
    \item monthly forecasts at the \emph{sort type} level (e.g., items that can be sorted by a machine vs those which cannot); and
    \item daily sort-type forecasts.
\end{enumerate}
All forecasts are collected over an $18$-month horizon, generated on a forecast snapshot date in the recent past. The finest-grain forecast of weekly SKUs makes up the majority of the entries of the vector $\yhat$. This leads to the following aggregations: 

\begin{enumerate}
    \item In each month, daily sort-type forecasts sum to monthly sort-type forecasts. 
    \item Daily sort-type forecasts and daily product-group forecasts must agree at the aggregate level.
    \item Within each product family, monthly sort-type forecasts must agree with the monthly product-family forecast.
    \item Weekly SKU-level forecats must agree with day-to-week aggregations at the product group level.
\end{enumerate}

To study the feasibility of reconciling the demand forecasts using the methods described in Section \ref{sec:algo}, we rely on high memory cloud computing infrastructure: specifically, AWS EC2 spot instance of type \texttt{u7i-8tb.112xlarge} with 8TB of memory.  The implementations rely on Python \cite{python_lang_ref}, Pandas \cite{mckinney-proc-scipy-2010} and Scipy's \texttt{scipy.sparse} linear algebra routines \cite{2020SciPy-NMeth}. 

The constraint matrix $\mathbf{A}$ can be derived from standard tabular datasets using the \texttt{pandas} \texttt{groupby} method on a the concatenated tabular forecasts datasets.  For details on construction of $\mathbf{A}$, see Sec. \ref{sec:constraints}.  In our large scale experiments $\mathbf{A}$ has $47,213$ rows, representing $47,213$ constraint equations, and $4,185,173,500$ columns (the length of $\mathbf{y}$). We rely on sparse matrix implementations in Python's \texttt{scipy.sparse} package.  In our experiment, constructing matrix $\mathbf{A}$ took less than an hour.  

To construct the diagonal weight matrix $\mathbf{W}$, for each forecast $y_i$, we assign its weight $[\mathbf{W}]_{i,i}=w_i/\widehat{y_i}^2$ to incorporate both its level of importance ($w_i$) and its scale ($1/\widehat{y_i}^2$). We further explain this setting in Section \ref{sec:discussion}. For this specific experiment, we assign $w_i=1$ to data with weekly SKU-level granularity, $w_i=1,000$ to daily sort-type and daily product-group forecasts, and $w_i=50,000$ to monthly sort-type and monthly product-family forecasts. The choice of weights depend on our belief and knowledge in these forecast models, and can also be tuned by running reconciliation on historical forecasts and comparing to the actual values.

\subsection{Constraints from Tabular Datasets \label{sec:constraints}}

Given several datasets that contain forecasts of various segments, how can one construct the constraint matrix $\mathbf{A}$? In this section we describe how to construct $\mathbf{A}$ directly from tabular datasets (avoiding, for example, manual specification), and mention requisite conditions on the tabular datasets.  In short, the procedure amounts to intersecting column names and grouping on the shared dimensions. 

Consider several tabular datasets, each with two types of columns: \emph{i)} dimensions (such as region, product classification, year or day), \emph{ii)} and metrics (such as units sold).  
Assume that columns of the datasets \emph{i)} share a common naming convention across the tabular datasets, are \emph{ii)} \emph{explicit} (for example, that `state' is a column if `county' is a column, and a second dataset includes `state') and \emph{iii)} each column \emph{partitions} the space (each column contains all the labels for that dimensions, and no duplication, so that aggregating out the column gives the correct total).  Under these three conditions, $\mathbf{A}$ can derived with a simple algorithm that intersects the columns of the tabular datasets and then groups-by the intersecting columns. 

We describe the approach for two tabular datasets and note it can be extended to several datasets by considering pairs. The approach proceeds as follows.   First, the two datasets are concatenated into a new dataset.   The vector of forecasts $\mathbf{y}$ corresponds to metric column from this concatenated dataset. 

Next, $\mathbf{A}$ is derived as follows.  The concatenated dataset is \emph{grouped-by} the shared dimensions (the columns that exist in both datasets). For each group key (a unique set of labels for the shared dimensions), the metric values of the rows from the first dataset must sum to equal the sum of metrics from the rows of second dataset by the assumptions described in the previous paragraphs. Let $\mathcal{I}_1$ be the indices of the rows from the first dataset, and $\mathcal{I}_2$ be the indices of the rows from the second dataset.  For a group key $k$, we set $A_{k,i} = 1 $ for $i \in \mathcal{I}_1$ and $A_{k,i} = -1$ for $i \in \mathcal{I}_2$ and $A_{k,i} = 0$ if $i$ is not in either set.

\subsection{Comparison of Different Algorithms}\label{sec:exp_comparison}

Before running the large-scale experiment, we compare the performance of algorithms mentioned in Section \ref{sec:algo} on a small reconciliation setting by only consider a subset of the forecasts (daily sort forecasts and monthly sort forecasts). For this problem, $\mathbf{A}\in\mathbb{R}^{3,384\times 106,408}$, meaning we have 3,384 constraints and $106,408$ forecasts.

In Table \ref{table:comparison1}, we compare the computing time of getting the LSQR solution (without the non-negativity constraint, as in Eq.~\eqref{eq:lsqr}), with algorithms including Alternating Projection, Dykstra's, and ADMM to incorporate the non-negativity constraint. We set $\epsilon_{\text{abs}}=10^{-7}$ and $\epsilon_{\text{rel}}=3\times10^{-8}$ for ADMM. Since $\epsilon_{\text{primal}}\approx \sqrt{n}\epsilon_{\text{abs}}$ with $\sqrt{n}$ around 300, we set $\epsilon_{\text{iter}}=\epsilon_{\text{fea}}=3\times 10^{-5}$ to roughly match their different stopping criteria. We also present the performance of Python's solver \texttt{CVXOPT}, which implements an interior-point algorithm, with $\epsilon_{\text{abs}}=10^{-7},\epsilon_{\text{rel}}=3\times 10^{-8},\epsilon_{\text{fea}}=3\times 10^{-5}$ too. 

\begin{table*}[ht]
\scriptsize
\centering
\begin{tabular}{|c|c|c|c|c|c|}
\hline
\textbf{Algorithm} & \textbf{Time} & $\frac{\|\mathbf{y}^*-\widehat{\mathbf{y}}\|}{\|\mathbf{y}^*\|}$ & $\|(\mathbf{y}^*)_-\|$ & $\|\mathbf{A}\mathbf{y}^*\|$ & \textbf{Note} \\ \hline
LSQR & 0.01s & 0.078 & $1\times10^{6}$ & $1.1\times10^{-6}$ &  \begin{tabular}[c]{@{}c@{}}Closed-form solution;\\ Does not consider the non-negativity constraint.\end{tabular}\\ \hline
\begin{tabular}[c]{@{}c@{}}Alternating\\ Projection\end{tabular} & 4.65s & 0.078 & $3\times10^{-5}$ & $6.2\times10^{-7}$ &  \begin{tabular}[c]{@{}c@{}}Simple and efficient;\\ Might not converge to the global optimal.\end{tabular}\\ \hline
Dykstra's & 5.30s & 0.078 & $3\times 10^{-5}$ & $1.1\times10^{-6}$ & \begin{tabular}[c]{@{}c@{}}Simple and efficient;\\ Always converges to the global optimal;\\ No guarantee for obj. optimality w. common stopping rules.\end{tabular} \\ \hline
ADMM & 11.60s & 0.027 & 121 & $9.7\times10^{-7}$ & \begin{tabular}[c]{@{}c@{}}Obj. suboptimality bound w. common stopping criteria;\\ Converges to modest accuracy quickly.\end{tabular}  \\ \hline
\begin{tabular}[c]{@{}c@{}}CVXOPT\\ (Interior-point)\end{tabular} & 217.93s & 0.078  & 0.0 & $1.1\times10^{-6}$ & \begin{tabular}[c]{@{}c@{}}Easy to implement in Python;\\ Designed for dense matrices; computationally expensive.\end{tabular} \\ \hline
\end{tabular}
\caption{Comparison Table for different algorithms on forecast reconciliation of size $106,408$. The third column, $\|\mathbf{y}^*-\widehat{\mathbf{y}}\|/ \|\mathbf{y}^*\|$, measures how far the reconciled forecasts are from the original.   The fourth and fifth columns measure if the reconciled forecasts violate the constraints.  }
\label{table:comparison1}
\end{table*}

The alternating projections method and Dykstra's algorithm perform similarly---both are efficient and their feasibility violations look trivial. ADMM takes a marginally longer time to converge given the above-mentioned tolerance levels. It gives a solution closer to the original forecasts, but its violation of the non-negativity constraint is larger than Dykstra's yet still acceptable given the size of $\mathbf{y}$. In Section \ref{sec:algo} we mention that \texttt{CVXOPT} is sensitive to poor scaling; indeed, based on this experiment, it's much slower than others.  Figure \ref{fig:comparison_small1} visualizes the results for one specific timeseries (additional plots are available in Appendix \ref{sec:app5}).

\begin{figure}[ht]
\centering
\includegraphics[width=\linewidth]{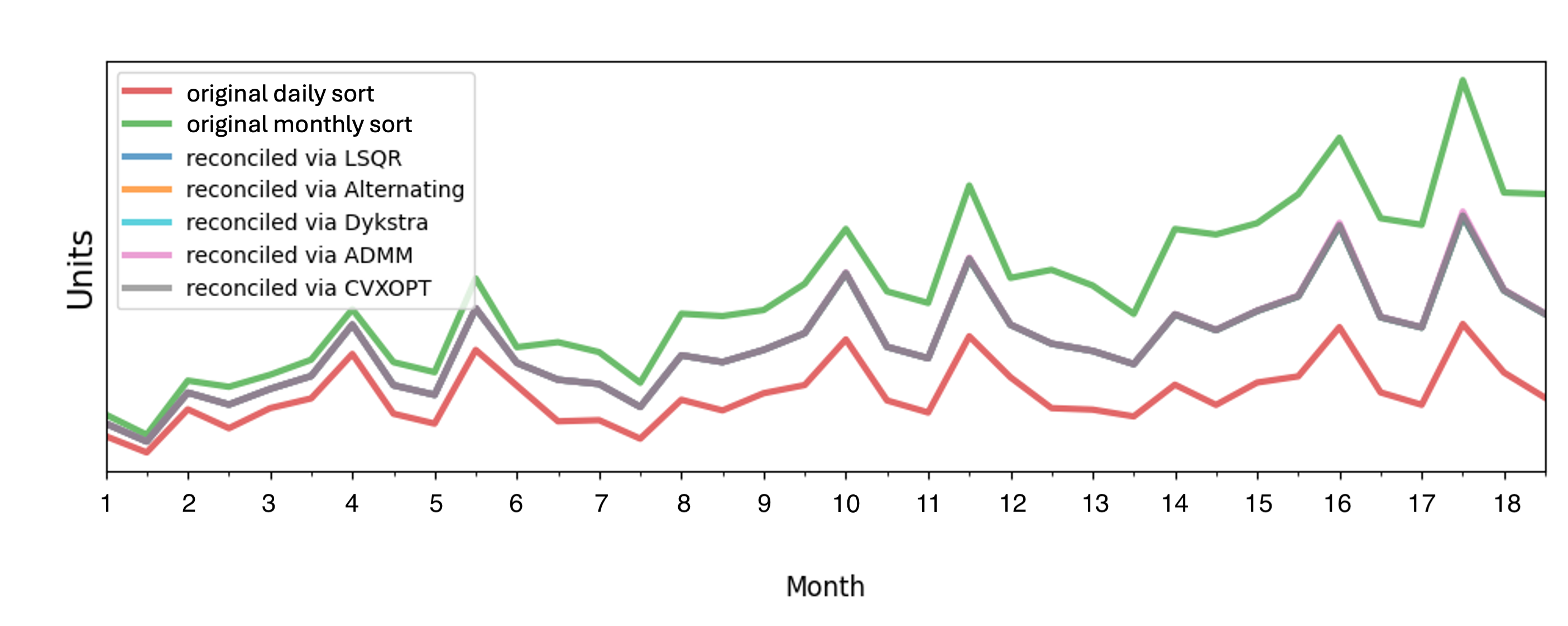}
\caption{Visualization of result for a problem size $106,408$.  All the algorithms, except LSQR, solve the same underlying problem and thus produce the same reconciled forecasts. LSQR does not have a non-negativity constraint, but still produces very similar results for this visualization. }
\label{fig:comparison_small1}
\end{figure}

\subsection{Implementation on Large-scale Forecasts}\label{sec:exp_large}

We use selected algorithms (LSQR, Alternating Projection, and Dykstra's) to complete large-scale forecast reconciliation; specifically, 
$\mathbf{y} \in \mathbb{R}^{4,185,173,500} \mbox{ and } \mathbf{A} \in \mathbb{R}^{47,213 \times 4,185,173,500}$.
For Alternating Projection method and Dykstra, we set $\epsilon_{\text{iter}}=1000$ (very small since $\|\widehat{\mathbf{y}}\|\propto10^{10}$) and $\epsilon_{\text{fea}}=10000$ (very small compared to $\|(\mathbf{y}^*_{\text{LSQR}})_{-}\|>3\times 10^6$). Their performances are summarized in Table \ref{table:comparison2}. We can see that within acceptable time, both Alternating Projection and Dykstra's algorithm generate optimal $\mathbf{y}^*$ that are close to the original forecast with minimal violation of the non-negativity constraint.

\begin{table}[ht]
\centering
\begin{tabular}{|c|c|c|c|c|}
\hline
\textbf{Algorithm} & \textbf{Time} & $\frac{\|\mathbf{y}^*-\widehat{\mathbf{y}}\|}{\|\mathbf{y}^*\|}$ & $\|(\mathbf{y}^*)_-\|$ \\ \hline
LSQR & 17.2 min & 0.096 & 5,438,210 \\ \hline
\begin{tabular}[c]{@{}c@{}}Alternating\\ Projection\end{tabular} & 388.1 min & 0.096 & 9,650 \\ \hline
Dykstra's & 489.3 min & 0.096 & 9,650 \\ \hline
\end{tabular}
\caption{Comparison Table for different algorithms on forecast reconciliation of size 4,185,173,500.}
\label{table:comparison2}
\end{table}

In Figure \ref{fig:comparison2_1}, we plot the original forecasts (except product-ID forecasts, which is weekly data), our reconciled values, and the actual units for one specific category, and we present another in Figure \ref{fig:comparison2_2} in Appendix \ref{sec:app5}. In each figure, the top subplot is the units over time, and the bottom one is the percentage difference to the actual values. We can see that the reconciled values lie between unreconciled forecasts most of the time. We also report the mean absolute percentage difference (MAPE) for each forecast in Table \ref{table:avg_error}. The daily sort and monthly sort forecasts are very close to the actual units, while the monthly family and daily product group forecasts are less accurate. The reconciled values from our algorithms lie in between. In fact, running this experiment with historical data can guide us to choose proper weights. For example, since daily sort and monthly sort forecasts are more accurate historically, we can assign more weights to them in our optimization framework. If we increase the weight of daily sort-fine-grained to $5,000$ (from $1,000$) and the weight of monthly sort to $500,000$ (from $50,000$), the reconciled values (displayed in brown dashed lines in Figure \ref{fig:comparison2_1} and \ref{fig:comparison2_2}) become closer to the actual values and we get smaller MAPEs as well. 

\begin{figure}[ht]
\centering
    \includegraphics[width=\linewidth]{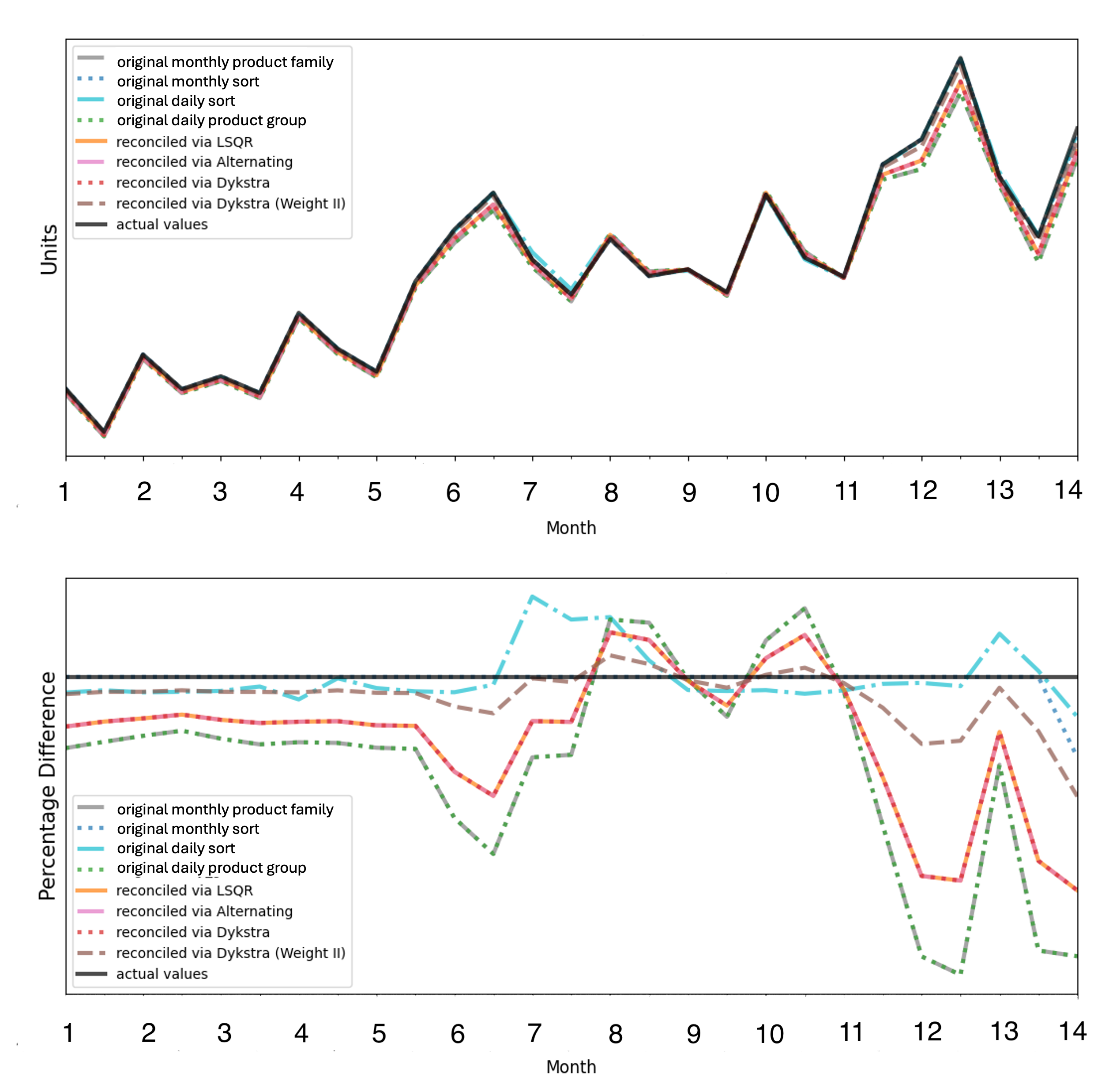}
\caption{Visualization of result for large-scale reconciliation for a segment of data. The top subplot is Units vs Month, while the bottom subplot is the percentage difference to the actual values vs Month.}
\label{fig:comparison2_1}
\end{figure}

\begin{table*}
\scriptsize
\centering
\begin{tabular}{|c|c|c|c|c|c|c|c|c|}
\hline
\textbf{Algorithm} & monthly product family    & monthly sort    & daily sort    & daily product group    & LSQR    & AP      & Dykstra & Dykstra w. Weight II \\ \hline
Segment 1 & 0.0120 & 0.0003 & 0.0024 & 0.0120 & 0.0082 & 0.0082 & 0.0082 & 0.0028\\ \hline
Segment 2 & 0.0194 & 0.0019 & 0.0030 & 0.0194 & 0.0126 & 0.0126 & 0.0126 & 0.0049\\ \hline
\end{tabular}
\caption{MAPE (Mean-absolute-percent-error) for original forecasts and results of reconciliation for for two segments.}
\label{table:avg_error}
\end{table*}

We also plot the reconciliation results between daily product-group and weekly SKU-level forecasts for one specific group in Figure \ref{fig:comparison_plot_1}. We present the result for another group in Appendix \ref{sec:app5}. We can see that our reconciled values lie between the two original forecasts, and for this product group, our reconciled values are quite close to the actual units.

\begin{figure}[ht]
\centering
\includegraphics[width=\linewidth]{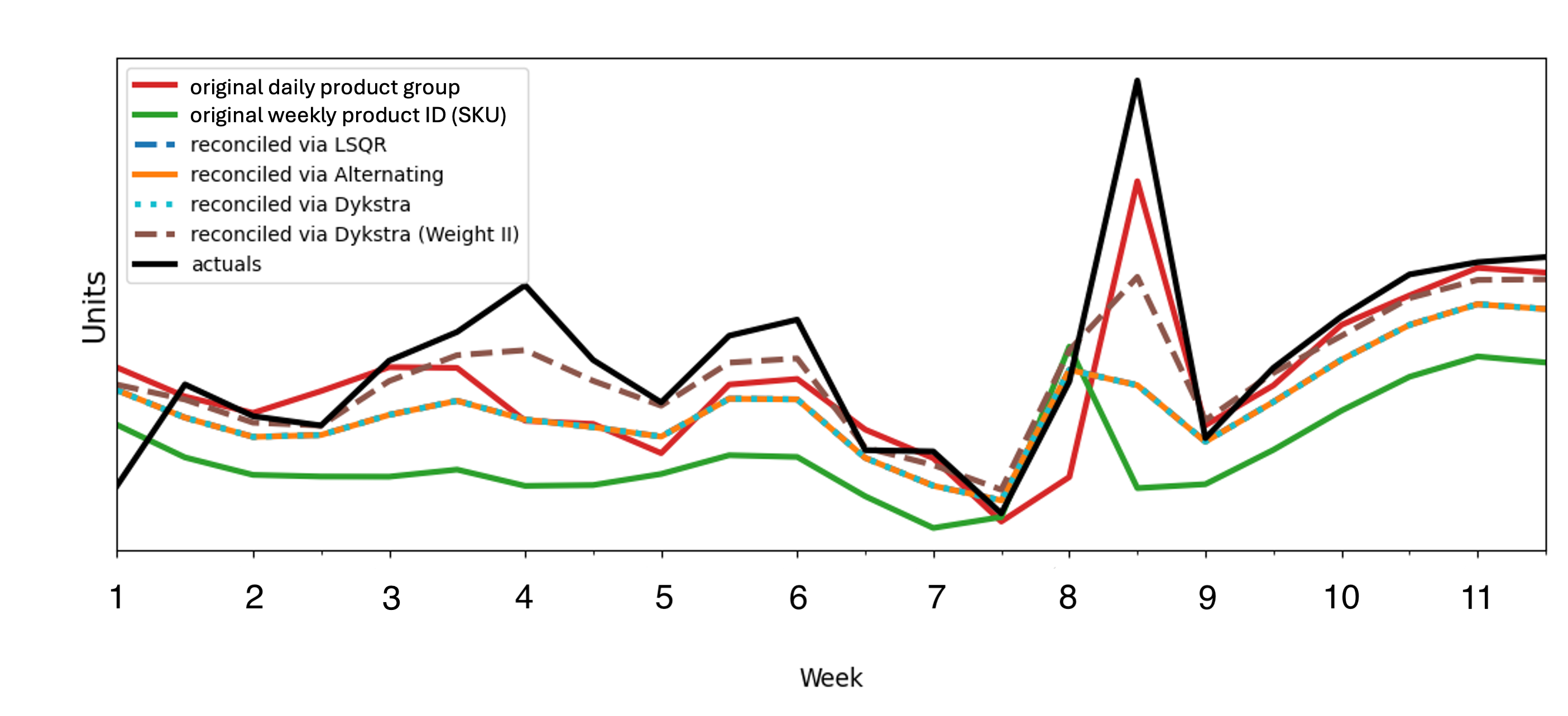}
\caption{Visualization of result for a problem size ($>$4 billion) for a segment of data.}
\label{fig:comparison_plot_1}
\end{figure}

\section{Discussion and Extensions}\label{sec:discussion}
\subsection{Absolute Error vs. Percentage Error}\label{sec:discussion_pct}

Because the scales of the forecasts to be reconciled (i.e. $\mathbf{y}$) differ greatly, minimizing $(\mathbf{y} - \mathbf{\widehat{y}})^\top\mathbf{W}(\mathbf{y}-\mathbf{\widehat{y}})$ can yield dramatic changes to small-scale forecasts. Motivated by the work of Davies \cite{davies2019_ext}, we consider minimizing the percentage loss $\|\frac{\mathbf{\widehat{y}}-\mathbf{y}}{\mathbf{\widehat{y}}}\|$ instead of $\|\mathbf{\widehat{y}}-\mathbf{y}\|$. The objective in Eq.~\eqref{eqn:basic} can be written as:
\begin{equation}
\begin{aligned} \label{eqn:basic_pct}
 \left(\frac{\mathbf{y} - \mathbf{\widehat{y}}}{\mathbf{\widehat{y}}+\epsilon}\right)^\top\mathbf{W}\left(\frac{\mathbf{y} - \mathbf{\widehat{y}}}{\mathbf{\widehat{y}}+\epsilon}\right)&=(\mathbf{y} - \mathbf{\widehat{y}})^\top\mathbf{\widehat{M}}\mathbf{W}(\mathbf{y}-\mathbf{\widehat{y}})\\
\end{aligned}
\end{equation}
where $\mathbf{\widehat{M}}$ is a diagonal matrix with diagonal $(\frac{1}{\mathbf{\widehat{y}}+\epsilon})^2$, and $\epsilon$ is a small fixed value to avoid $0$ in the denominator (in our experiments we set $\epsilon=1$). Although $\mathbf{\widehat{M}}$ and $\mathbf{W}$ have different roles---one takes care of the scale difference while the other incorporates our beliefs in different forecast models, mathematically we can treat $\mathbf{\widehat{M}}\mathbf{W}$ as a general weight matrix $\mathbf{W}$. For simplicity, we still use $(\mathbf{y} - \mathbf{\widehat{y}})^\top\mathbf{W}(\mathbf{y}-\mathbf{\widehat{y}})$ in the main draft.

\subsection{When is the least squares solution nonnegative?}\label{sec:discussion_share}

Although solving the large-scale optimization problem Eq.~\eqref{eqn:basic} with both equality and inequality constraints can be hard, we find that under certain conditions on the structure of the aggregation matrix $\mathbf{A}$ and weight matrix $\mathbf{W}$, the non-negativity inequality constraint is redundant. 

An aggregation matrix $\mathbf{A}$ has \emph{rows with disjoint supports} if there exists a single nonzero entry in each column of $\mathbf{A}$. This corresponds to  a two-level hierarchy where, for example, fine-grained regional forecasts are aggregated to the state level. In particular, a matrix with a single aggregation constraint, $\A = \begin{bmatrix} -1 & 1 & \cdots & 1 \end{bmatrix}$ representing a simple two-leveled hierarchy, has rows with disjoint support. The proof of the following theorems can be found in Appendix \ref{appendix:proof_nonnegativity} (see Corollary \ref{cor:disjoint_row_supports}).

\begin{theorem}\label{thm:nonneg}
Consider the forecast reconciliation problem \eqref{eqn:basic} with original forecast $\mathbf{\widehat{y}}> 0$. Assume $\mathbf{A}\in\mathbb{R}^{K\times N}$ has rows with disjoint supports and, the weight matrix $\mathbf{W}\in\mathbb{R}^{N\times N}$ is diagonal with $W_{nn}=1/\widehat{y}_n$. If there exists an optimal solution satisfying $\mathbf{A}\mathbf{y}=\mathbf{0}$, then the least squares solution to Eq.~\eqref{eq:lsqr} is positive, and thus is the optimal solution to Eq.~\eqref{eqn:basic}.  
\end{theorem}

From this theorem we learn that if setting weights $W_{nn}=1/\widehat{y}_n$ is a reasonable choice and if $\mathbf{A}$ has rows with disjoint supports---albeit a strong restriction---we can simply drop the non-negativity constraint in the optimization problem and use the closed form solution Eq.~\eqref{eq:lsqr}. We note that Theorem \ref{thm:nonneg} does not hold for arbitrary aggregation matrices $\mathbf{A}$, and provide a counterexample in Appendix \ref{appendix:counterexample}. Despite this, we found experimentally that in many cases the least squares solution was nonnegative. 

Next we consider two variations on the $W_{nn} = 1/\widehat{y}_n$ weighting scheme. We will say the rows of $\A$ are indexed by aggregation \emph{constraints} and the columns are indexed by forecasted \emph{items}. We think of $\A$ as representing aggregations in a hierarchical structure $\hier$. We assume the hierarchy $\hier$ is \emph{tree-based}, meaning it does not have cycles (for example, if $\yhat_n$ represents a forecast of item $I$ in state $S$, then $I$ could be aggregated to a product category and $S$ could be aggregated to a region or both could be aggregated to total regional sales. In a \emph{tree-based} hierarchy each item has at most one parent. A tree-based hierarchy is \emph{strict} if each parent item is the aggregation of exactly one set of child items. Tree-based hierarchies are sufficient to support top-down disaggregation; strict tree-based hierarchies are sufficient to support buttom-up aggregation.

The \emph{depth} of item $n$, denoted $D(n)$ is defined recursively so that a top-level item with no parent has depth $0$ and the depth of any other item is $1$ more than that of its (necessarily unique) parent. Conversely, the \emph{height} of item $n$, denoted $H(n)$ is defined by $H(n) = D_{\max} - D(n)$, where $D_{\max}$ is the maximum depth of any item in $\hier$. In particular this implies that $H(c) = H(p) - 1$ if $c$ is a child of $p$. If $\A$ is the aggregation matrix of a tree-based hierarchy, we can assume wlog that the rows of $\A$ are ordered so that $H(p_1) \geq H(p_2) \geq \cdots \geq H(p_K)$. See Figure \ref{fig:tree_example} in Appendix \ref{appendix:weighting_proofs} for an example.

For a constant $M$, let $\W$ be the diagonal matrix with $w_n = M^{H(n)} / \yhat_n$, meaning the weights are inversely proportional to the value of the initial prediction $\yhat_n$ and scale exponentially in proportion to the height. We call $\W$ the \emph{top-heavy} weighting of $\hier$. Similarly, we can define the \emph{bottom-heavy} weighting of $\hier$ by setting $w_n = M^{D(n)} / \yhat_n$. 

The following result shows that in the limit, the least squares solutions to the top-heavy and bottom-heavy weighting schemes converge, respectively, to the top-down (share-based) and bottom-up reconciliation solutions with the top-level (resp., bottom-level) items fixed by $\yhat$. The proof can be found in Appendix \ref{appendix:weighting_proofs}. This is important as it means the optimization based approach can be viewed as a flexible extension of traditional top-down and bottom-up forecasting. It also implies that including the non-negativity constraint in Eq.~\eqref{eqn:basic} is redundant in the sense that the solution to the unconstrained optimization is \emph{always} non-negative in this setting.

\begin{theorem}\label{thm:top_bottom_heavy_aggs}
Let $\hier$ be a tree-based hierarchy with aggregation matrix $\A$.
    \begin{enumerate}
        \item  Let $\W$ be the top-heavy weighting of $\hier$. Let $\y^*(M)$ be the corresponding least-squares solution to Eq.~\eqref{eq:lsqr}. Then $\lim\limits_{M \rightarrow \infty} \y^*(M)$ exists and converges to the share-based disaggregation over $\hier$ where $\lim\limits_{M \rightarrow \infty} y^*_n(M) = \yhat_n$ for any top-level node $n$ and the shares are proportional to the values of $\yhat$ (see Theorem \ref{thm:top_level_limits} and Corollary \ref{cor:top_down_convergence}).
        \item Assume $\hier$ is strict and let $\W$ be the bottom-heavy weighting of $\hier$. Then $\lim\limits_{M \rightarrow \infty} \y^*(M)$ exists and converges to the bottom-up aggregation over $\hier$ where $\lim\limits_{M \rightarrow \infty} y^*_n(M) = \yhat_n$ for any bottom-level node $n$ (see Corollary \ref{cor:depth_based_convergence}).
    \end{enumerate}
\end{theorem}

One immediate extension is to set all weights to be the reciprocal of the corresponding forecasted value, including the weight of the aggregate forecast, and we then have the weight setting in Theorem \ref{thm:nonneg}. On the other hand, this discussion indicates one advantage of using our optimization approach, which allows for more weighting possibilities. For example, as discussed in Section \ref{sec:discussion}, one can also set the weight as $(1/\widehat{\mathbf{y}})^2$ to consider the percentage error. In short, with our optimization formulation, we can set the weights in any way we like.

\section{Related Work}\label{sec:related}
Forecast reconciliation methods have evolved through numerous formulations, as documented in various studies (e.g., \cite{athanasopoulos,wickramasuriya}) and we refer the reader those survey articles.  The formulation in (\ref{eqn:basic}) was studied for reconciliation of financial data \cite{byron1978estimation}. The central importance of forecast reconciliation is certainly not new.  Our works specifically adopts a formulation that frames forecast reconciliation as a constrained quadratic programming problem. This approach is particularly valuable as it accommodates arbitrary linear constraints, enabling both multiple hierarchies and non-hierarchical constraints. While much work has been developed to address this class of problems (\cite{pang,dykstra,boyd}), our focus is on algorithms suitable for industry-scale applications. Given that existing packages often face limitations with larger problem sizes, we explore methods specifically designed for massive-scale implementations. 

While a number of works discuss scaling limitations of forecast reconciliation, to the best of our knowledge, our work considers problems that are orders of magnitude larger than those published previously (for example, see \cite{SPRANGERS20241689} which discussed scaling into the thousands).  Comparison is challenging since much work does not consider temporal hierarchies, which significantly reduces the scale of the problems.  For example, the public M5 dataset, the dimension is 12,350 as formulations do not include temporal constraints and the problem can be solved independently for each time period, and direct or out of the box optimization packages work.

\bibliographystyle{abbrv}
\bibliography{citation_external}

\appendix
\section*{Appendix}

\section{Aggregation matrices with disjoint row supports}\label{sec:app3} \label{appendix:proof_nonnegativity}
Throughout this section we fix the following notation:
\begin{itemize}
    \item $\mathbf{A} = (a_{k,n}) \in \{0,1,-1\}^{K \times N}$ is an aggregation matrix;
    \item  $A_1,\ldots,A_K$ are the row vectors of $\mathbf{A}$;
    \item $\mathbf{W}$ is a diagonal weight matrix with diagonal entries $\mathbf{W}_{n,n}:=w_n > 0$ for $1 \leq n \leq N$;
    \item $\mathbf{\widehat{y}} > \mathbf{0}$ is an initial (column) forecast vector;
    \item $\mathbf{y}^*_{\text{LSQR}}$ is the least squares solution from Eq.~\eqref{eqn:basic}.
\end{itemize}

\begin{theorem}\label{thm:weighted_lsqr_entries}
    Assume $\mathbf{A}$ has rows with disjoint supports. For any $1 \leq m \leq N$, 
    \begin{equation} \label{eq:weighted_lsqr_soln}
        (\mathbf{y}^*_{\text{LSQR}})_m = \mathbf{\widehat{y}}_m - \frac{A_{\ell}\mathbf{\widehat{y}}}{A_{\ell} \mathbf{W}^{-1} A_\ell^{\top}} \frac{a_{\ell,m}}{w_{m}},
    \end{equation}
    where $\ell$ is the unique row index for which the $m$th column of $\mathbf{A}$ is nonzero. 
\end{theorem}

\begin{proof}
    Because the rows of $\mathbf{A}$ have disjoint support and $\W$ is diagonal, $\mathbf{A}\mathbf{W}^{-1}\mathbf{A}^{\top}$ is diagonal.

    Therefore, for any $1 \leq k \leq K$,
    $$
    ((\mathbf{A}\mathbf{W}^{-1}\mathbf{A}^{\top})^{-1}\mathbf{A}\mathbf{\widehat{y}})_k = \frac{A_k\mathbf{\widehat{y}}}{A_k\mathbf{W}^{-1} A_k^{\top}},
    $$
    treating the $1 \times 1$ matrices in the numerator and denominator as scalars. This is because 
    \begin{enumerate}
        \item the diagonal entries of $\mathbf{A}\mathbf{W}^{-1}\mathbf{A}^{\top}$ are $A_k \mathbf{W}^{-1} A_k^{\top}$,
        \item the entries of $\mathbf{A}\mathbf{\widehat{y}}$ are $A_k\mathbf{\widehat{y}}$, and
        \item $(\mathbf{A}\mathbf{W}^{-1}\mathbf{A}^{\top})^{-1}$ acts on the vector $\mathbf{A}\mathbf{\widehat{y}}$ by scalar multiplication of its diagonal entries.
    \end{enumerate} 

    The columns of $\mathbf{W}^{-1}\mathbf{A}^{\top}$ are $\mathbf{W}^{-1} A_k^{\top}$ for $1 \leq k \leq K$, and thus (because multiplying a matrix by a vector expands as a linear combination of the columns of the matrix),
    \begin{equation} \label{eq:lin_combo_expression}
    \mathbf{W}^{-1}\mathbf{A}^{\top}(\mathbf{A}\mathbf{W}^{-1}\mathbf{A}^{\top})^{-1}\mathbf{A}\mathbf{\widehat{y}} = \sum_{k=1}^K\frac{A_k\mathbf{\widehat{y}}}{A_k\mathbf{W}^{-1} A_k^{\top}}\mathbf{W}^{-1} A_k^{\top}.
    \end{equation}

    Now consider the $m$th entry of this resultant vector. Because $(A_k^{\top})_m = a_{k,m}$ is only nonzero when $k = \ell$ and $\mathbf{W}^{-1}$ is diagonal, it follows that $(\mathbf{W}^{-1} A_k^{\top})_m$ is nonzero only when $k = \ell$, in which case its value is $\frac{a_{\ell,m}}{w_{m}}$. Therefore, Eq.~\eqref{eq:lin_combo_expression} collapses to a single summand and
    $$
    (\mathbf{y}^*_{\text{LSQR}})_m = \mathbf{\widehat{y}}_m - \frac{A_{\ell}\mathbf{\widehat{y}}}{A_{\ell} \mathbf{W}^{-1} A_\ell^{\top}} \frac{a_{\ell,m}}{w_{m}}.
    $$    
\end{proof}

\begin{corollary} \label{cor:disjoint_row_supports}
    Assume $\mathbf{A}$ has rows with disjoint supports and $w_n=1/\mathbf{\widehat{y}}_n$ for all $1 \leq n \leq N$. Then $\mathbf{y}^*_{\text{LSQR}} \geq \mathbf{0}$. 
\end{corollary}

\begin{proof}
    As in Theorem \ref{thm:weighted_lsqr_entries}, for any $1 \leq m \leq N$ let $1 \leq \ell \leq K$ be the unique row index where the $m$th column of $\mathbf{A}$ is nonzero.  Plugging $1/w_m = \mathbf{\widehat{y}}_m$ into Eq.~\eqref{eq:weighted_lsqr_soln} gives 
    \begin{align*}
    (\mathbf{y}^*_{\text{LSQR}})_m &= \mathbf{\widehat{y}}_m - \frac{A_{\ell}\mathbf{\widehat{y}}}{A_{\ell} \mathbf{W}^{-1} A_\ell^{\top}} a_{\ell,m}\mathbf{\widehat{y}}_m \\
    &= \frac{\mathbf{\widehat{y}}_m}{A_\ell \mathbf{W}^{-1} A_\ell^{\top}}\left(A_\ell \mathbf{W}^{-1} A_\ell^{\top} - a_{\ell,m} A_\ell \mathbf{\widehat{y}}\right).
    \end{align*}
    Because $\mathbf{W}^{-1}$ is diagonal with nonnegative entries it is positive semidefinite, meaning $A_\ell \mathbf{W}^{-1} A_\ell^{\top} \geq 0$. Thus it suffices to prove 
    $$
    A_\ell \mathbf{W}^{-1} A_\ell^{\top} - a_{\ell,m} A_\ell \mathbf{\widehat{y}} \geq 0.
    $$
    We compute
    \begin{align*}
    A_\ell \mathbf{W}^{-1} A_\ell^{\top} - a_{\ell,m} A_\ell \mathbf{\widehat{y}} 
    &= \sum_{n=1}^Na_{\ell,n}^2\mathbf{\widehat{y}}_n - a_{\ell,m} \sum_{n=1}^Na_{\ell,n}\mathbf{\widehat{y}}_n \\
    &= (1-a_{\ell,m})\sum_{n\, : \, a_{\ell,n} = 1} \mathbf{\widehat{y}}_n  \\
    &+ (1 + a_{\ell,m})\sum_{n\, : \, a_{\ell,n} = -1} \mathbf{\widehat{y}}_n.
    \end{align*}
    The resulting quantity is nonnegative because $a_{\ell,m} = \pm 1$.
\end{proof}

\begin{corollary}
    If $\mathbf{A}$ consists of a single row and $w_n = 1/\mathbf{\widehat{y}}_n$ for all $1 \leq n \leq N$, then $\mathbf{y}^*_{\text{LSQR}} \geq \mathbf{0}$.
\end{corollary}
\begin{proof}
    A matrix with a single row vacuously has disjoint row supports.
\end{proof}

\section{Proofs with top-heavy and bottom-heavy weighting}\label{appendix:weighting_proofs}
If $\hier$ is a tree-based hierarchy, each of its constraints can be written in the form
\begin{equation} \label{eq:constraint_rep}
\y_{p_k} - \sum_{c \in C_k} \y_{c} = 0,
\end{equation}
for $1 \leq k \leq K$. Here $p_k$ is a \emph{parent} item and $C_k$ is its set of \emph{child} items. By Theorem \ref{thm:representation} we can assume the aggregation matrix of any tree-based hierarchy is written according to this \emph{canonical representation}.

The \emph{depth} of item $n$, denoted $D(n)$ is defined recursively so that a top-level item with no parent has depth $0$ and the depth of any other item is $1$ more than that of its (necessarily unique) parent. Conversely, the \emph{height} of item $n$, denoted $H(n)$ is defined by $H(n) = D_{\max} - D(n)$, where $D_{\max}$ is the maximum depth of any item in $\hier$. In particular this implies that $H(c) = H(p) - 1$ if $c$ is a child of $p$. If $\A$ is the aggregation matrix of a tree-based hierarchy, we can assume wlog that the rows of $\A$ are ordered so that $H(p_1) \geq H(p_2) \geq \cdots \geq H(p_K)$.

\begin{example}
Figure \ref{fig:tree_example} shows an example of a strict tree-based hierarchy $\hier$ and its canonical aggregation matrix. Each level of $\hier$ is labeled with the depth ($D$) and height ($H$) of all items on that level. The items in this example are $\{a,\ldots,g\}$ and the columns of $\A$ are indexed in that same order. For example $p_2 = b$ and $C_2 = \{d,e\}$, and the second constraint is $\y_b - \y_d - \y_e = 0$. 
\end{example}

\newcommand{\mytikzfig}{%
\begin{tikzpicture}[
    level/.style={sibling distance=40mm/##1},
    level distance=15mm,
    scale = .8,
]
\node (root) {$a$}
    child {node (y2) {$b$}
        child {node (y4) {$d$}}
        child {node (y5) {$e$}}
    }
    child {node (y3) {$c$}
        child {node (y6) {$f$}}
        child {node (y7) {$g$}}
    };

\draw[dotted] (-3mm,0) -- (-4,0) node[anchor = east] {$\substack{D=0 \\ H=2}$};
\draw[dotted] (-23mm,-15mm) -- (-4,-15mm) node[anchor = east] {$\substack{D=1 \\ H=1}$};
\draw[dotted] (-33mm,-30mm) -- (-4,-30mm) node[anchor = east] {$\substack{D=2 \\ H=0}$};
\end{tikzpicture}
}

\begin{figure}[htb]
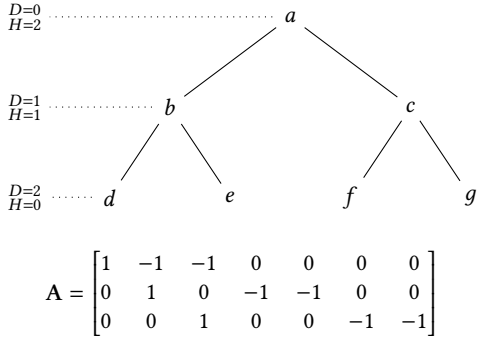
 
\centering
\makeatletter
\if@twocolumn
    \mytikzfig
    
\ \\

    $
    \A = \begin{bmatrix}
    1 & -1 & -1 & 0 & 0 & 0 & 0 \\ 
    0 & 1 & 0 & -1 & -1 & 0 & 0 \\
    0 & 0 & 1 & 0 & 0 & -1 & -1 \\ 
    \end{bmatrix}
    $
\else
    \begin{minipage}{0.45\linewidth}
        \mytikzfig
    \end{minipage}
    \begin{minipage}{0.45\linewidth}
        \vspace{-15mm}$
    \A = \begin{bmatrix}
    1 & -1 & -1 & 0 & 0 & 0 & 0 \\ 
    0 & 1 & 0 & -1 & -1 & 0 & 0 \\
    0 & 0 & 1 & 0 & 0 & -1 & -1 \\ 
    \end{bmatrix}
    $
    \end{minipage}
\fi
\makeatother

    \caption{A strict tree-based hierarchy and its canonical aggregation matrix.}
    \label{fig:tree_example}
\end{figure}

\begin{proposition}
A tree-based hierarchy $\hier$ can be represented by an aggregation matrix $\A$ of full rank in which each column has at most one $-1$ entry. If in addition, $\hier$ is strict, then each column of $\A$ has at most one $+1$ entry as well. 
\end{proposition}
\begin{proof}
Let $P = \{p_1,\ldots,p_K\}$ be the parent nodes in $\hier$ and consider the restriction of $\A$ to the columns indexed by $P$. Assuming the rows are ordered so that $H(p_1) \geq H(p_2) \geq \cdots \geq H(p_K)$, we can permute the columns so they are indexed correspondingly in order as $p_1,\ldots,p_K$. The resulting restriction has $+1$'s on its diagonal and is upper-triangular because any child comes after its parent in order. Therefore this restriction is full-rank and hence so too is $\A$. 

The fact that each column has at most one $-1$ entry corresponds to the fact that each item has at most one parent. The fact that each column has at most one $+1$ entry when $\hier$ is strict corresponds to the fact that each item can only be a parent in a single aggregation constraint. 
\end{proof}

\subsection{Notation and definitions} \label{section:notation}

In order to solve the optimization problem in Eq.~\eqref{eqn:basic}, we consider the Lagrangian relaxation

\begin{align}
\label{eq:lagr_def} \lagr &= \frac{1}{2} || \y - \yhat||^2_{\W} + \sum_{k=1}^K \lambda_k A_k\yhat \\
\nonumber &= \frac{1}{2} \sum_{n=1}^N w_n(\y_n-\yhat_n)^2 + \sum_{k=1}^K \lambda_k A_k\yhat.
\end{align}

Without the nonnegativity constraint $\y \geq \zeros$, the solution to the optimization problem in Eq.~\eqref{eqn:basic} occurs at points where $\nabla \lagr(M)= \zeros$. Because the loss function is convex, there is a unique point where $\nabla \lagr(M)= \zeros$ and it must be a global minimum.  Dropping the nonnegativity constraint will not be a problem here as our goal is to study cases where $\y$ corresponds to a share-based disaggregation or bottom-up aggregation, and in both of those cases nonnegativity will automatically be satisfied. 

\begin{lemma} \label{lemma:lagr_soln}
For any $\A$ and $\W$, the solution to $\nabla \lagr = \zeros$ from Eq.~ \eqref{eq:lagr_def} is 
$$
\begin{bmatrix} \y \\ \blam \end{bmatrix} := \begin{bmatrix} \yhat-\W\inv\A^{\top}(\A\W\inv\A^{\top})\inv \A\yhat \\ (\A\W\inv\A^{\top})\inv \A\yhat\end{bmatrix}
$$
\end{lemma}
\begin{proof}
Note that $\nabla \lagr = \zeros$ is equivalent to solving 
$$
\begin{bmatrix} \W & \A^{\top} \\ \A & \zeros \end{bmatrix}
\begin{bmatrix} \y \\ \blam \end{bmatrix} = 
\begin{bmatrix} \W\yhat \\ \zeros \end{bmatrix}
$$
and perform block row reduction. 
\end{proof}

We state the next two lemmas without proof as they are simple but useful algebraic manipulations. 
\begin{lemma} \label{lemma:lagr_zero}
Let $\A$ be the canonical aggregation matrix of a tree-based hierarchy $\hier$, and let $\W$ be a diagonal weight matrix. For any $1 \leq n \leq N$, let $P_n \subseteq \{1,\ldots,K\}$ be the (possibly empty) set of rows of $\A$ where item $n$ is a parent and let $C_n \subseteq \{1,\ldots,N\}$ be the (possibly empty) set of rows where item $n$ is a child. Then $\frac{\partial\lagr}{\partial\y_n} = 0$ if and only if 
\begin{equation}\label{eq:partial_lagr_zero}
\y_n = \yhat_n - \frac{1}{w_n}\sum_{k \in P_n} \lambda_k + \frac{1}{w_n}\sum_{j \in C_n} \lambda_j. 
\end{equation}
\end{lemma}

The following result is unmotivated for now but will be useful in later proofs. 

\begin{lemma}
    Let $\bL = \diag(\ell_1,\ldots,\ell_K)$ be a diagonal matrix of weights. For any $\A$ and diagonal $\W$, extract $\blam$ from the solution to $\nabla \lagr = \zeros$ and consider 
    $$
    \bL\inv \blam  = \bL\inv (\A\W\inv\A^{\top})\inv\A\yhat = (\A\W\inv\A^{\top}\bL)\inv\A\yhat.
    $$
    Then the entries of the matrix $\Q:= \A\W\inv\A^{\top}\bL$ are
\begin{equation}\label{eq:Q_structure}
\Q_{i,j} = (\A\W\inv)_{i,:}(\A^{\top}\bL)_{:,j} = \sum_{n=1}^N a_{i,n}a_{j,n}\frac{\ell_j}{w_n}.
\end{equation}
\end{lemma}

\subsection{Share-based disaggregation and top-heavy weighting}

For this section we will assume $\A$ comes from a tree-based hierarchy $\hier$. Recall the notation that $p_k$ and $C_k$ are, respectively, the parent item and its children imposed by the $k$th constraint and that $H(n)$ is the height of the $n$th item. 

For a constant $M$, let $\W$ be the diagonal matrix with $w_n = M^{H(n)} / \yhat_n$, meaning the weights are inversely proportional to the value of the initial prediction $\yhat_n$ and scale exponentially in proportion to the height. We call $\W$ the \emph{top-heavy} weighting of $\hier$. Our goal is to prove that as $M \rightarrow \infty$, (1) the least squares solution  $\y^*_{\lsqr}$ to Eq.~\eqref{eq:lsqr} converges, and (2) in the limit the top-level items agree with $\yhat$ while lower-level items are found through share-based disaggregation. A tree-based hierarchy is necessary for share-based disaggregation to be well-defined without adding further constraints. 

Throughout this section, for a fixed $M$ we will let $\y(M) = \y = (\y_n)_{n=1}^N$ and $\blam(M) = \blam = (\lambda_k)_{k=1}^K$ be the solutions to $\nabla \lagr = \zeros$ from Lemma \ref{lemma:lagr_soln}.

\begin{theorem} \label{thm:height_lambdas_converge_zero}
Let $\hier$ be a tree-based hierarchy with canonical aggregation matrix $\A$ and top-heavy weight matrix $\W$. For all $1 \leq k \leq K$, 
\begin{align*}
    \lim\limits_{M \rightarrow \infty} \frac{\lambda_k}{M^{H(p_k)}} = 0.
\end{align*} 
\end{theorem}

\begin{proof}
Fix $M$ and let $\bL$ be the $K \times K$ diagonal matrix with diagonal entries $M^{H(p_k)}$. Consider $\bL\inv \blam$, whose entries are the $\frac{\lambda_k}{M^{H(p_k)}}$ that we want to understand. Let $\Q = \A\W\inv\A^{\top}\bL$ so that $\bL\inv \blam = \Q\inv \A \y$ and consider the entries $\Q_{i,j}$. 

By Eq.~\eqref{eq:Q_structure}, $\Q$ has diagonal entries
\begin{align*}
    \Q_{i,i} &= \frac{\yhat_{p_i}}{M^{H(p_i)}}M^{H(p_i)} + \sum_{c \in C_i} \frac{\yhat_c}{M^{H(c)}}M^{H(p_i)} \\
    &= \yhat_{p_i} + M \sum_{c \in C_i} \yhat_c.
\end{align*}
On the other hand, if $i \neq j$, then $\Q_{i,j} = 0$ unless there exists some item $1 \leq n \leq N$ such that $a_{i,n}$ and $a_{j,n}$ are nonzero. Because $\A$ is tree-based, if such an $n$ exists, then it is unique and this can only happen if
\begin{enumerate}
    \item[Case 1:] item $n$ is a child on row $i$ and a parent on row $j$, meaning $i<j$ and $H(n) = H(p_i)-1 = H(p_j)$;
    \item[Case 2:] item $n$ is a parent on row $i$ and a child on row $j$, meaning $i>j$ and $H(n) = H(p_i) = H(p_j)-1$; or
    \item[Case 3:] item $n$ is a parent on both row $i$ and row $j$, meaning $H(n) = H(p_i) = H(p_j)$.
\end{enumerate} 
Therefore, if $\Q_{i,j} \neq 0$ then by Eq.~\eqref{eq:Q_structure}
\begin{align*}
    \Q_{i,j} &= -\frac{\yhat_n}{M^{H(n)}}\cdot M^{H(p_j)} 
    = \begin{cases} 
    -\yhat_n & \text{ in Case 1,} \\
    -M\yhat_n & \text{ in Case 2, and }\\
    -\yhat_n & \text{ in Case 3.}
    \end{cases}
\end{align*}

This tells us the diagonal entries of $\Q$ are  linear in $M$, the entries above the diagonal of $\Q$ are constant (Case 1 or Case 3), and the entries below the diagonal of $\Q$ are constant or linear in $M$ (Case 2 or Case 3). Now consider $\Q\inv$, which we will write as the product of $\frac{1}{\det(\Q)}$ with its cofactor matrix. 

First we claim $\det(\Q)$ is a polynomial in $M$ of degree $K$. Express $\det(\Q) = \sum_{\sigma} (-1)^{\text{sgn}(\sigma)} \prod_{i} \Q_{i,\sigma(i)}$, where the sum is over all permutations of $\{1,\ldots,K\}$. When $\sigma$ is the identity permutation, the product of the diagonal entries of $\Q$ contributes an $M^K$ term. For any non-identity permutation $\sigma$ there is some $i$ where $\sigma(i) > i$ and hence $\Q_{i,\sigma(i)}$ is a constant that is independent of $M$. Therefore each of the non-identity terms contribute polynomials of degree at most $K-1$, and the $M^K$ term from the diagonal cannot be cancelled.

Similarly each minor of $\Q$ is a polynomial of degree at most $K-1$ because it is the determinant of a $(K-1) \times (K-1)$ matrix whose entries are at most linear in $M$. Therefore every entry in the cofactor matrix is a polynomial of degree at most $K-1$. This means $\Q\inv \rightarrow \zeros_{K \times K}$ as $M \rightarrow \infty$.  

It follows that $\bL\inv\blam = \Q\inv \A\yhat \rightarrow \zeros$ as well because $\A\yhat$ is independent of $M$. 
\end{proof}

\begin{corollary} \label{cor:height_lambdas_converge}
Let $\hier$ be a tree-based hierarchy with canonical aggregation matrix $\A$ and top-heavy weight matrix $\W$. For all $1 \leq k \leq K$, 
$$
\lim\limits_{M \rightarrow \infty} \frac{\lambda_k}{M^{H(p_k)-1}}
$$
exists and is finite. 
\end{corollary}
\begin{proof}
    Note that because $\blam = (\A\W\inv\A^{\top})\inv\A\yhat$, each of its entries is a rational function of $M$, meaning (away from a finite set of values of $M$ where it is undefined) each $\lambda_k$ can be written as the ratio of two polynomials in $M$. 
    
    Let $f$ and $g$ be polynomials such that $\lambda_k = \frac{f(M)}{g(M)}$. The fact that $\frac{\lambda_k}{M^{H(p_k)}} = \frac{f(M)}{M^{H(p_k)}g(M)}$ converges to $0$ as $M \rightarrow \infty$ means the degree of $f(M)$ is less than the degree of $M^{H(p_k)}g(M)$. Thus the degree of $f(M)$ is less than or equal to the degree of $M^{H(p_k)-1}g(M)$, meaning $$\lim\limits_{M \rightarrow \infty}\frac{\lambda_k}{M^{H(p_k)-1}} = \lim\limits_{M \rightarrow \infty}\frac{f(M)}{M^{H(p_k)-1}g(M)}$$ exists and is finite. 
\end{proof}

\begin{theorem} \label{thm:top_level_limits}
Let $\hier$ be a tree-based hierarchy with canonical aggregation matrix $\A$ and top-heavy weight matrix $\W$. If $p$ is a top-level parent node in $\hier$, then 
\begin{align*}
    \y^*_{p}:= \lim\limits_{M \rightarrow \infty} \y_{p}(M) = \yhat_{p}.
\end{align*}
\end{theorem}
\begin{proof}
When $p$ is a top-level node, it does not have a parent. Let $I_p \subseteq \{1,\ldots,K\}$ be the set of constraints where $p$ is the parent node.  By Lemma \ref{lemma:lagr_zero}, 
\begin{align*}
\y_{p}&= \yhat_{p}-\frac{\yhat_{p}}{M^{H(p)}} \sum_{k \in I_p} \lambda_k.
\end{align*}
By Theorem \ref{thm:height_lambdas_converge_zero}, $\frac{\lambda_k}{M^{H(p)}} \rightarrow 0$ for each $k \in I_p$ and therefore $\y_p \rightarrow \yhat_p$ as $M \rightarrow \infty$.
\end{proof}

\begin{theorem} \label{thm:child_limits}
Let $\hier$ be a tree-based hierarchy with canonical aggregation matrix $\A$ and top-heavy weight matrix $\W$. Let $c$ be an item in $\hier$ that is not a top-level parent node and let $p_k$ be its parent. Then
\begin{equation} \label{eq:height_lambdas_limit}
\lim\limits_{M \rightarrow \infty} \frac{\lambda_k}{M^{H(p_k)-1}} = \lim\limits_{M \rightarrow \infty} \frac{\y_c}{\yhat_c}-1.
\end{equation}
In particular, $\y^*_c:= \lim\limits_{M \rightarrow \infty} \y_c$ exists and is finite.
\end{theorem}

\begin{proof}
We need only establish Eq.~\eqref{eq:height_lambdas_limit}. Assuming this, Corollary \ref{cor:height_lambdas_converge} implies the limit on the left converges, and therefore the limit on the right does as well. 

    Let $J_c \subseteq \{1,\ldots,K\}$ be the (possibly empty) set of constraints where $c$ is the parent. By Lemma \ref{lemma:lagr_zero}, 
    \begin{align*}
        \y_c &= \yhat_c + \frac{\yhat_c}{M^{H(c)}} \lambda_k - \frac{\yhat_c}{M^{H(c)}}\sum_{j \in I_c}\lambda_j \\
        \frac{\y_c}{\yhat_c} - 1  &= \frac{\lambda_k}{M^{H(c)}} - \sum_{j \in J_c} \frac{\lambda_j}{M^{H(c)}} \\
          &= \frac{\lambda_k}{M^{H(p_k) - 1}} - \sum_{j \in J_c} \frac{\lambda_j}{M^{H(p_j)}}.
    \end{align*}
    By Corollary \ref{cor:height_lambdas_converge} $\lim\limits_{M \rightarrow \infty} \frac{\lambda_k}{M^{H(p_k)-1}}$ converges and by Theorem \ref{thm:height_lambdas_converge_zero} $\lim\limits_{M \rightarrow \infty} \frac{\lambda_j}{M^{H(p_j)}} = 0$ for all $j \in I_c$ . This establishes Eq.~\eqref{eq:height_lambdas_limit}.
\end{proof}

\begin{corollary} \label{cor:top_down_convergence}
Let $\hier$ be a tree-based hierarchy with canonical aggregation matrix $\A$ and top-heavy weight matrix $\W$. Then $\y^* = \lim\limits_{M \rightarrow\infty} \y(M)$ exists and corresponds to share-based disaggregation. In particular, for any $1 \leq k \leq K$ and any $s \in C_k$, 
$$
\y^*_s = \frac{\yhat_s}{\sum_{c \in C_k} \yhat_c}\y^*_{p_k}.
$$
This says $\y^*_s$ is the share-based disaggregation of its parent $\y^*_{p_k}$ based on the weights from $\yhat$. 
\end{corollary}

\begin{proof}

Every item in $\hier$ is either a top-level parent item or a child.  Theorem \ref{thm:top_level_limits} and Theorem \ref{thm:child_limits} establish the convergence of $\y^*_n$ for each item $1 \leq n \leq N$. 

For each $1 \leq k \leq K$, the constraint $\frac{\partial\lagr}{\partial\lambda_k} = 0$ says $\y_{p_k} = \sum_{c \in C_k} \y_c$ and therefore by convergence, $\y^*_{p_k} = \sum_{c \in C} \y^*_c$.

Fix $s \in C_k$.  By Eq.~\eqref{eq:height_lambdas_limit}, for any $c \in C_k$ we have $\frac{\y^*_s}{\yhat_s} = \frac{\y^*_c}{\yhat_c}$ because $\lim\limits_{M \rightarrow \infty} \frac{\lambda_k}{M^{H(p_k)-1}}$ is a finite constant that is independent of $c$. Therefore $\y^*_c = \frac{\y^*_s}{\yhat_s} \yhat_c$ (note this holds vacuously when $c=s$) and
\begin{align*}
\y^*_{p_k}  = \sum_{c \in C_k} \y^*_c &= \sum_{c \in C_k} \frac{\y^*_s}{\yhat_s} \yhat_c \\
\yhat_s \y^*_{p_k} &= \y^*_s \sum_{c \in C_k} \yhat_c \\
\frac{\yhat_s}{\sum_{c \in C_k} \yhat_c} \y^*_{p_k} &= \y^*_s.
\end{align*}
\end{proof}

\subsection{Bottom-up aggregation and bottom-heavy weighting}

We begin with a depth-based analogue of Theorem \ref{thm:height_lambdas_converge_zero} and Corollary \ref{cor:height_lambdas_converge}.
\begin{theorem}\label{thm:depth_lambdas_converge}
Let $\hier$ be a strict tree-based hierarchy with aggregation matrix $\A$ and bottom-heavy weight matrix $\W$. For all $1 \leq k \leq K$, 
\begin{align*}
\lim\limits_{M \rightarrow \infty} \frac{\lambda_k}{M^{D(p_k) + 1}} = 0.
\end{align*}
Consequently, 
\begin{align*}
\lim\limits_{M \rightarrow \infty} \frac{\lambda_k}{M^{D(p_k)}}
\end{align*}
exists and is finite.
\end{theorem}
\begin{proof}
The proof is very similar to that of Theorem \ref{thm:height_lambdas_converge_zero}. 

Fix $M$ and let $\bL$ be the $K \times K$ diagonal matrix with diagonal entries $M^{D(p_k)+1}$. Consider $\bL\inv \blam$, whose entries are the $\frac{\lambda_k}{M^{H(p_k)}}$ that we want to understand. Let $\Q = \A\W\inv\A^{\top}\bL$ so that $\bL\inv \blam = \Q\inv \A \y$ and consider the entries $\Q_{i,j}$. 

On the diagonal, by Eq.~\eqref{eq:Q_structure},
\begin{align*}
    \Q_{i,i} &= \frac{\yhat_{p_i}}{M^{D(p_i)}}M^{D(p_i)+1} + \sum_{c \in C_i} \frac{\yhat_c}{M^{D(p_i)+1}}M^{D(p_i)+1} \\
    &= M\yhat_{p_i} + \sum_{c \in C_i} \yhat_c.
\end{align*}
On the other hand, if $i \neq j$, then $\Q_{i,j} = 0$ unless there exists some item $1 \leq n \leq N$ such that $a_{i,n}$ and $a_{j,n}$ are nonzero. Because $\hier$ is tree-based, if such an $n$ exists, then it is unique and  because $\hier$ is strict, this can only happen if\footnote{This is the place where the proof of Theorem \ref{thm:height_lambdas_converge_zero} does not carry over verbatim. If $n$ were allowed to be a parent on both row $i$ and row $j$, then we would recover $\Q_{i,j} = -M\yhat_n$. In would then be possible that the minors of $\Q$ could be polynomials of degree $K$, which breaks the rest of the argument as we could not guarantee $\Q\inv \rightarrow \zeros$.}
\begin{enumerate}
    \item[Case 1:] item $n$ is a child on row $i$ and a parent on row $j$, meaning $i<j$ and $D(n) = D(p_i)+1 = D(p_j)$ or
    \item[Case 2:] item $n$ is a parent on row $i$ and a child on row $j$, meaning $i>j$ and $D(n) = D(p_i) = D(p_j)+1$.
\end{enumerate} 
Therefore, if $\Q_{i,j} \neq 0$ again by Eq.~\eqref{eq:Q_structure}
\begin{align*}
    \Q_{i,j} &= -\frac{\yhat_n}{M^{D(n)}}\cdot M^{D(p_j)+1} 
    = \begin{cases} -M\yhat_n & \text{ if } i < j \\ -\yhat_n & \text{ if } i > j.\end{cases}
\end{align*}

The remainder of the proof is identical to that of Theorem \ref{thm:height_lambdas_converge_zero}, with the exception that in computing $\det(\Q)$, it is now the case that for any non-identity permutation $\sigma$ there is some $i$ where $\sigma(i) < i$ and hence $\Q_{i,\sigma(i)}$ is a constant that is independent of $M$.
\end{proof}

\begin{corollary} \label{cor:depth_based_convergence}
Let $\hier$ be a strict tree-based hierarchy with aggregation matrix $\A$ and bottom-heavy weight matrix $\W$. Then
\begin{enumerate}
    \item $\y^*:= \lim\limits_{M \rightarrow \infty} \y(M)$ exists and is finite, 
    \item for each $1 \leq k \leq K$, $\y^*_{p_k} = \sum_{c \in C_k} \y^*_{c}$ and 
    \item if $c$ is a bottom-level item with no children, then $\y^*_c = \yhat_c$. 
\end{enumerate}
In other words, $\y^*$ corresponds to bottom-up aggregation. 
\end{corollary}

\begin{proof}
First we will prove $\y^*_{p_k}$ exists. Let $J_k \subseteq \{1,\ldots,K\}$ be the set of constraints on which $p_k$ is a child ($J_k$ is either empty or a singleton). By Lemma \ref{lemma:lagr_zero},
\begin{align*}
    \y_{p_k} &= \yhat_{p_k} -  \frac{\yhat_{p_k}}{M^{D(p_k)}}\lambda_k + \frac{\yhat_{p_k}}{M^{D(p_k)}}\sum_{j \in J_k}\lambda_j \\
    &= \yhat_{p_k} \left(1 - \frac{\lambda_k}{M^{D(p_k})} + \sum_{j \in J_k} \frac{\lambda_j}{M^{D(p_j)+1}}\right).
\end{align*}
By Theorem \ref{thm:depth_lambdas_converge}, the limit of the quantity on the right side converges as $M \rightarrow \infty$. 

On the other hand, if $c$ has no children and parent $p_j$, then by Lemma \ref{lemma:lagr_zero}, 
\begin{align*}
\y_{c} &= \yhat_{c}-  \frac{\yhat_c}{M^{D(c)}}\lambda_j = \yhat_{c} \left(1 + \frac{\lambda_k}{M^{D(p_j)+1}}\right).
\end{align*}
By Theorem \ref{thm:depth_lambdas_converge} $\frac{\lambda_k}{M^{D(p_j)+1}} \rightarrow 0$ as $M \rightarrow \infty$ and hence $\y_c \rightarrow \yhat_c$. 

Thus we have proved the existence of $\y^*$ and that $\y^*_c = \yhat_c$ for bottom-level children. All other items are parents and the constraint $\y_{p_k} = \sum_{c \in C_k} \y_c$ passes to the limit, implying $\y^*_{p_k} = \sum_{c \in C_k} \y^*_{c}$ for all $k$. This completes the proof. 
\end{proof}

\section{Operations on the aggregation matrix}
\begin{theorem}\label{thm:representation}
Let $\mathbf{A} \in \mathbb{R}^{K \times N}$ be an aggregation matrix, let $\mathbf{W} \in \mathbb{R}^{N \times N}$ be a weight matrix, and let $\mathbf{E} \in \mathbb{R}^{K \times K}$ be invertible. If the least squares solution $$\widehat{\mathbf{y}} - \mathbf{W}^{-1}\mathbf{A}^{\top}(\mathbf{A}\mathbf{W}^{-1}\mathbf{A}^{\top})^{-1}\mathbf{A}\widehat{\mathbf{y}},$$ is nonnegative then so is the least squares solution $$\widehat{\mathbf{y}} - \mathbf{W}^{-1}\mathbf{B}^{\top}(\mathbf{B}\mathbf{W}^{-1}\mathbf{B}^{\top})^{-1}\mathbf{B}\widehat{\mathbf{y}},$$ where $\mathbf{B} = \mathbf{E}\mathbf{A}$.
\end{theorem}

\begin{proof}
We can compute 
\begin{align*}
& \mathbf{W}^{-1}\mathbf{B}^{\top}(\mathbf{B}\mathbf{W}^{-1}\mathbf{B}^{\top})^{-1}\mathbf{B} \\
=& \mathbf{W}^{-1}\mathbf{A}^{\top}\mathbf{E}^{\top}(\mathbf{E}\mathbf{A}\mathbf{W}^{-1}\mathbf{A}^{\top}\mathbf{E}^{\top})^{-1}\mathbf{E}\mathbf{A} \\
=& \mathbf{W}^{-1}\mathbf{A}^{\top}\mathbf{E}^{\top}(\mathbf{E}^{\top})^{-1}(\mathbf{A}\mathbf{W}^{-1}\mathbf{A}^{\top})^{-1}\mathbf{E}^{-1}\mathbf{E}\mathbf{A} \\
=& \mathbf{W}^{-1}\mathbf{A}^{\top}(\mathbf{A}\mathbf{W}^{-1}\mathbf{A}^{\top})^{-1}\mathbf{A}.
\end{align*}
\end{proof}

\begin{corollary}
    If $\mathbf{A}$ is row equivalent to a matrix with disjoint row supports and $W_{nn} = 1/\widehat{y}_n$, then  the the least squares solution Eq.~\eqref{eqn:basic} is nonnegative.
\end{corollary}

\section{A General Counterexample to Theorem \ref{thm:nonneg}} \label{appendix:counterexample}

\begin{example}\label{ex:nonnegative_counterexample}
Let $$\mathbf{A} = \begin{bmatrix} 1 & 0 & -1 & 0 & -1 \\ 0 & 1 & 0 & -1 & -1 \end{bmatrix}$$
and $$\widehat{\mathbf{y}} = \begin{bmatrix} 1, 1, 5, 5, 1\end{bmatrix}^{\top}.$$

Then setting weights $w_n = 1/\widehat{y}_n$ gives the following weighted least squares solution to Eq.~\eqref{eq:weighted_lsqr_soln}:
$$
\begin{bmatrix} 1.625 & 1.625 & 1.875 & 1.875 & -0.25\end{bmatrix}^{\top}.
$$
Therefore Theorem \ref{thm:weighted_lsqr_entries} cannot be generalized for arbitrary aggregation matrices and arbitrary initial solutions $\widehat{\mathbf{y}}$.

\end{example}

\section{Additional plots for Section \ref{sec:exp}}\label{sec:app5}
Figure \ref{fig:comparison_small2} presents the reconciliation result for another specific group (in addition to Figure \ref{fig:comparison_small1}). Figure \ref{fig:comparison2_2} presents the result from large-scale optimization for another category (in addition to Figure \ref{fig:comparison2_1}). Figure \ref{fig:comparison_plot_2} presents the reconciliation results between daily product group forecasts and weekly product ID (SKU) forecasts for another group (in addition to Figure \ref{fig:comparison_plot_1}).

\begin{figure}[h]
\centering
\includegraphics[width=\linewidth]{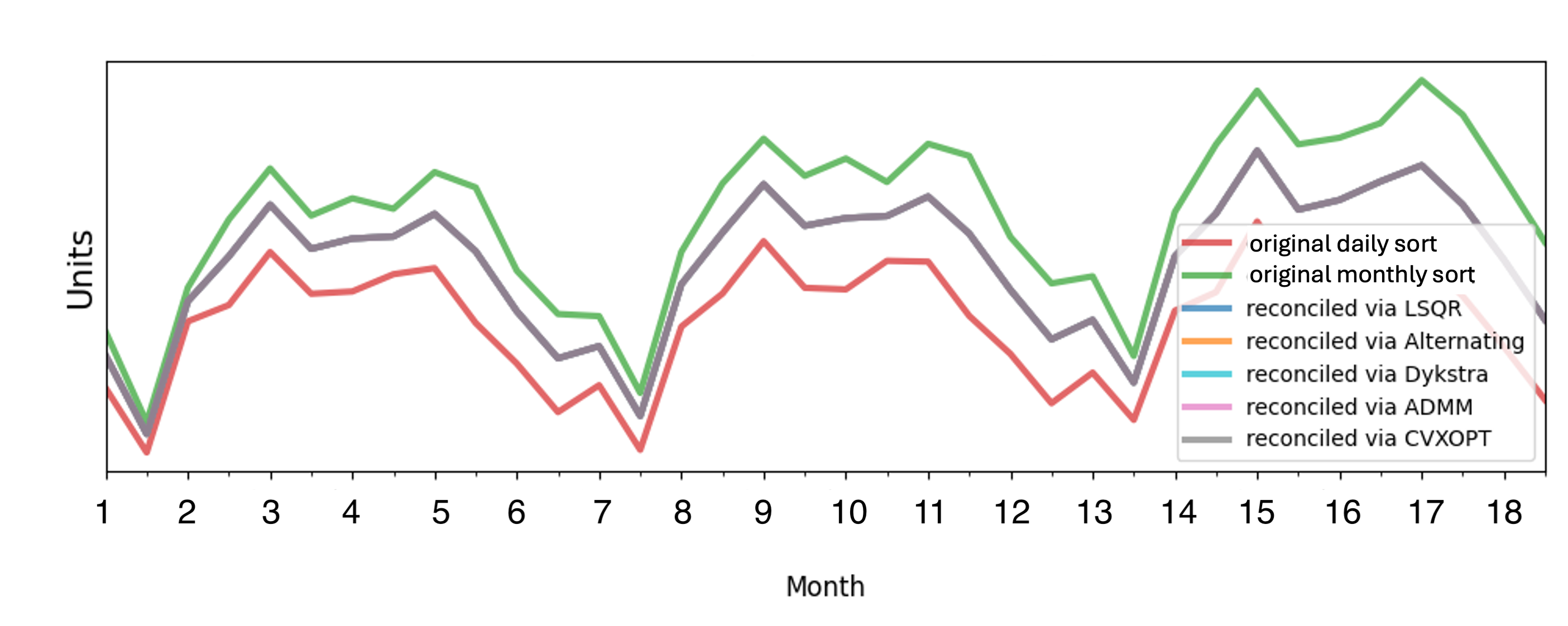}
\caption{Visualization of result for a problem size $106,408$ for another group, in addition to Figure \ref{fig:comparison_small1}.}
\label{fig:comparison_small2}
\end{figure}

\begin{figure}[b]
\centering
\includegraphics[width=\linewidth]{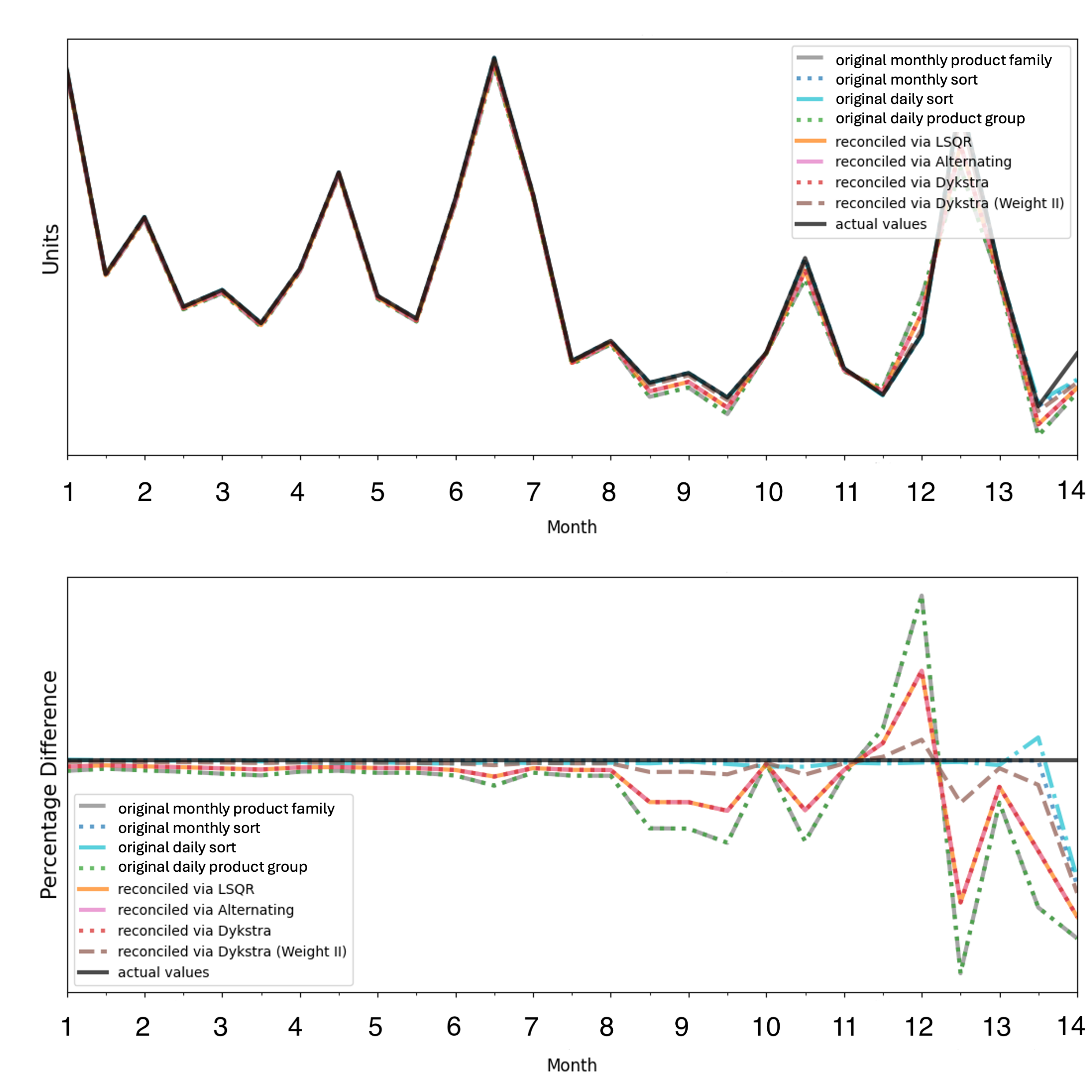}
\caption{Visualization of result for a problem size ($>$4 billion) for a segment of data. The top plot is Units vs Month, while the bottom subplot is the percentage difference to actual values vs Month.}
\label{fig:comparison2_2}
\end{figure}

\begin{figure}[b]
\centering
\includegraphics[width=\linewidth]{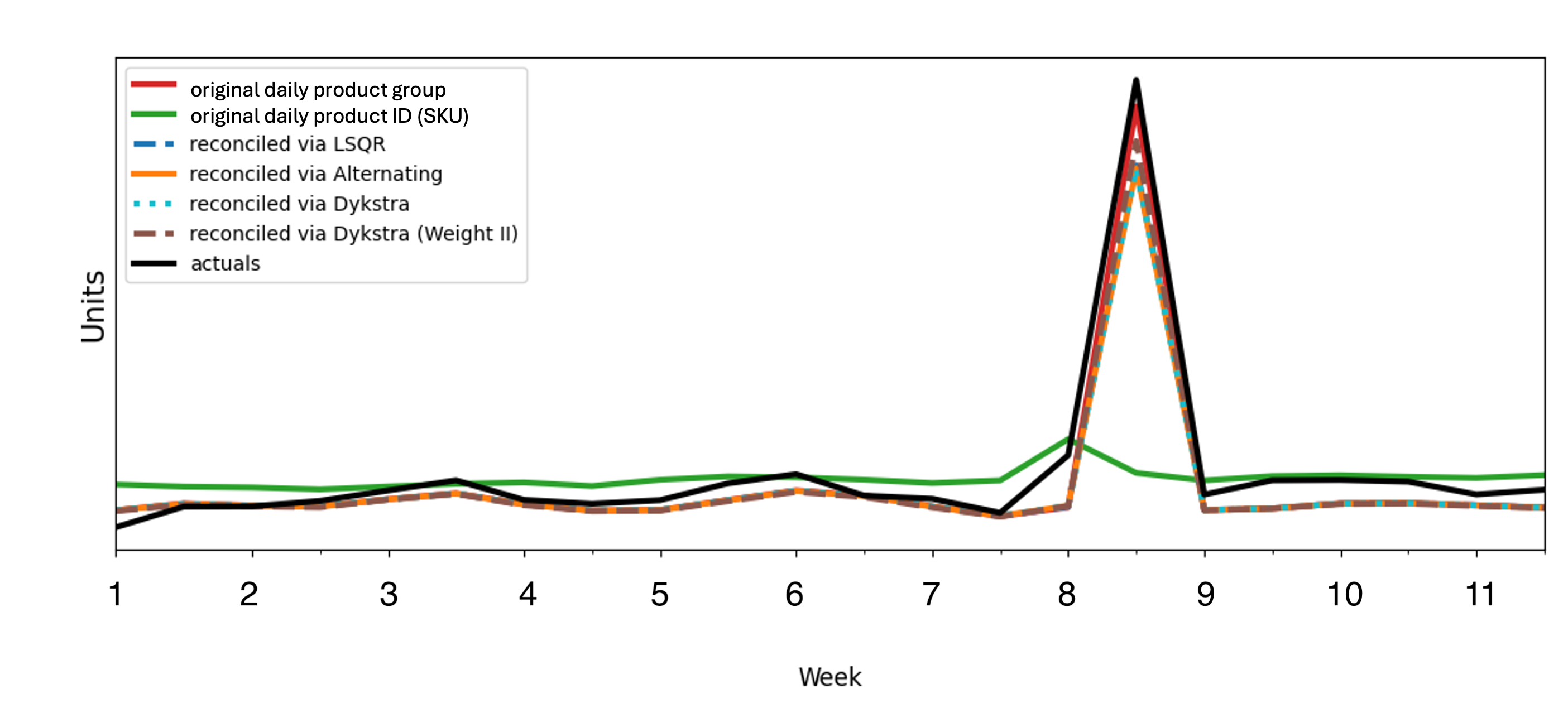}
\caption{Visualization of result for a problem size ($>$4 billion) for a segment of data.}
\label{fig:comparison_plot_2}
\end{figure}

\end{document}